\documentclass[12pt]{article}
\usepackage[utf8]{inputenc}
\usepackage{amsthm,amsmath,amssymb,bbm,hyperref,graphicx,float,tikz}
\usepackage{xcolor}
\allowdisplaybreaks
\usepackage[normalem]{ulem} 

\newtheorem{thm}{Theorem}
\newtheorem{prop}{Proposition}
\newtheorem{lemma}{Lemma}
\newtheorem{cor}{Corollary}
\theoremstyle{definition}
\newtheorem{rmk}{Remark}

\theoremstyle{remark}

\DeclareMathOperator{\tr}{tr}
\DeclareMathOperator{\Var}{Var}
\newcommand{\RE}{\textup{Re\,}}
\newcommand{\IM}{\textup{Im\,}}
\newcommand{\Hilbert}{\mathcal{H}}
\newcommand{\CCC}{\mathbb{C}}
\newcommand{\EEE}{\mathbb{E}}
\newcommand{\NNN}{\mathbb{N}}
\newcommand{\PPP}{\mathbb{P}}
\newcommand{\RRR}{\mathbb{R}}
\newcommand{\SSS}{\mathbb{S}}
\newcommand{\be}{\begin{equation}}
\newcommand{\ee}{\end{equation}}
\newcommand{\pr}[1]{|#1\rangle\langle #1 |}
\newcommand{\scp}[2]{\langle #1|#2 \rangle}

\newcommand{\GAP}{\textup{GAP}}
\newcommand{\mc}{\mathrm{mc}}
\newcommand{\can}{\mathrm{can}}
\newcommand{\GA}{\textup{GA}}

\title{Canonical Typicality For Other Ensembles Than Micro-Canonical}
\author{
Stefan Teufel\footnote{Mathematics Institute, Eberhard Karls University T\"ubingen, 
	Auf der Morgenstelle 10, 72076 T\"ubingen, Germany. ORCID: 0000-0003-3296-4261,
	E-mail: stefan.teufel@uni-tuebingen.de},~
Roderich Tumulka\footnote{Mathematics Institute, Eberhard Karls University T\"ubingen, 
	Auf der Morgenstelle 10, 72076 T\"ubingen, Germany. 
	ORCID: 0000-0001-5075-9929
	E-mail: roderich.tumulka@uni-tuebingen.de},~ 
Cornelia Vogel\footnote{Mathematics Institute, Eberhard Karls University T\"ubingen, 
	Auf der Morgenstelle 10, 72076 T\"ubingen, Germany.
	ORCID: 0000-0002-3905-4730,
	E-mail: cornelia.vogel@uni-tuebingen.de}
}

\begin{document}

\maketitle

\begin{abstract}
We generalize L\'evy's lemma, a  concentration-of-measure result for the uniform probability distribution on high-dimensional spheres, to a much more general class of measures, so-called GAP measures. For any given density matrix $\rho$ on a separable Hilbert space $\Hilbert$, $\GAP(\rho)$ is the most spread out probability measure on the unit sphere of $\Hilbert$ that has density matrix $\rho$ and thus forms the natural generalization of the uniform distribution. We prove concentration-of-measure whenever the largest eigenvalue $\|\rho\|$ of $\rho$ is small. We use this fact to generalize and improve well-known and important typicality results of quantum statistical mechanics to GAP measures, namely canonical typicality and dynamical typicality. Canonical typicality is the statement that for ``most'' pure states $\psi$ of a given ensemble, the reduced density matrix of a sufficiently small subsystem is very close to a $\psi$-independent matrix. Dynamical typicality is the statement that for any observable and any unitary time-evolution, for ``most'' pure states $\psi$ from a given ensemble the (coarse-grained) Born distribution of that observable in the time-evolved state $\psi_t$ is very close to a $\psi$-independent distribution. So far, canonical typicality and dynamical typicality were known for the uniform distribution on finite-dimensional spheres, corresponding to the micro-canonical ensemble, and for  rather special mean-value ensembles.  Our result shows that these typicality results hold also for $\GAP(\rho)$, provided the density matrix $\rho$ has small eigenvalues. Since certain GAP measures are quantum analogs of the canonical ensemble of classical mechanics, our results can also be regarded as a version of equivalence of ensembles.

\medskip

\end{abstract}

\section{Introduction}

In the 21st century, a modern perspective on quantum statistical mechanics is to consider an individual closed system in a pure state and investigate its and its subsystems' thermodynamic behavior; see, e.g., \cite{GM03,GMM04,PSW06,Reimann08b, BG09,GLMTZ09,GLMTZ10,GLTZ10,Short11,SF12,GHT13,GHT15,Reimann2015, GogEis16,BRGSR18,Reimann2018a,Reimann2018b,RG20,TTV22-physik,SWGW22} after pioneering work in \cite{vonNeumann29,Schroe52, Deutsch91,Srednicki94,Tasaki98}.

Roughly speaking, ``canonical typicality'' is the statement that the reduced density matrix of a subsystem obtained from a pure state of the total system is nearly deterministic if the pure state is randomly drawn from a sufficiently large subspace and the subsystem is not too large. 
More precisely, the original statement of canonical typicality \cite{Lloyd,GM03,PSW06,GLTZ06} asserts that for most pure states $\psi$ from a high-dimensional (e.g., micro-canonical) subspace $\Hilbert_R$ of the Hilbert space $\Hilbert_S$ of a macroscopic quantum system $S$ and for a subsystem $a$ of $S=a\cup b$ so that $\Hilbert_S=\Hilbert_a \otimes \Hilbert_b$,
the reduced density matrix
\be\label{rhosdef}
\rho_a^\psi:= \tr_b |\psi\rangle \langle \psi|
\ee
is close to the partial trace of $\rho_R:=P_R/d_R$ (the normalized projection to $\Hilbert_R$) and thus deterministic, provided that $d_R:=\dim \Hilbert_R$ is sufficiently large: 
\be\label{can1}
\rho_a^\psi \approx \tr_b \rho_R\,.
\ee
Here, the words ``most $\psi$'' refer to the \emph{uniform distribution} $u_R$ (normalized surface area measure) over the unit sphere
\be\label{SSSdef}
\SSS(\Hilbert_R) := \{\psi\in\Hilbert_R: \|\psi\|=1\}
\ee
in $\Hilbert_R$.
The name ``canonical typicality'' comes from the fact that if $\Hilbert_R=\Hilbert_\mc$ is a micro-canonical subspace and thus $\rho_R=\rho_\mc$ a micro-canonical density matrix, then $\tr_b \rho_\mc$ is close to the \emph{canonical density matrix}
\be\label{rhoacan}
\rho_{a,\can} = \frac{1}{Z_a}e^{-\beta H_a}
\ee
for $a$ with suitable $\beta$, provided $b$ is large and the interaction between $a$ and $b$ is weak; see, e.g., \cite{GLTZ06} for a summary of the standard derivation of this fact.

In this paper, we replace the uniform distribution by other, much more general distributions, so-called GAP measures, and show that for them a \emph{generalized canonical typicality} remains valid. 
For any density matrix $\rho$ replacing $\rho_R$ in $\Hilbert_S$, $\GAP(\rho)$ is the most spread-out distribution over $\SSS(\Hilbert_S)$ with density matrix $\rho$; the acronym stands for \emph{Gaussian adjusted projected} measure \cite{JRW94,GLTZ06b}. For $\rho=\rho_\mathrm{can}$, it arises as the distribution of wave functions in thermal equilibrium \cite{GLTZ06b,GLMTZ15}. If a system is initially in thermal equilibrium for the Hamiltonian $H_0$ but then driven out of equilibrium by means of a time-dependent $H_t$, its wave function will still be $\GAP(\rho)$-distributed for suitable $\rho$.
For general $\rho$, we think of $\GAP(\rho)$ as the natural ensemble of wave functions with density matrix $\rho$; for a more detailed description, see Section~\ref{sec:background}.  

We prove quantitative bounds asserting that for any $\rho$ with small eigenvalues (so $\rho$ is far from pure) and $\GAP(\rho)$-most $\psi\in\SSS(\Hilbert_S)$,
\be\label{gencan1}
\rho_a^\psi \approx \tr_b \rho\,.
\ee

Some reasons for seeking this generalization are: first, that it is mathematically natural; second, that in situations in which we can ask what the actual distribution of $\psi$ is (more detail later), this distribution might not be uniform; third, that it shows that the sharp cut-off of energies involved in the definition of $\Hilbert_\mc$ actually plays no role; and finally, that it informs and extends our picture of the equivalence of ensembles. A more detailed discussion of these reasons is given in Section~\ref{sec:moti}. 

As a direct consequence of generalized canonical typicality let us mention that, just as canonical typicality implies that for most pure states $\psi\in \SSS(\Hilbert_S)$ 
the entanglement entropy $-\tr(\rho_a^\psi \log \rho_a^\psi)$ has nearly the maximal value $\log d_a$ with $d_a=\dim\Hilbert_a$ \cite{hayden2006} (because $\rho_a^\psi\approx \tr_b I_S/D = I_a/d_a$ with $I$ the identity operator and $D=d_ad_b=\dim\Hilbert_S$), generalized canonical typicality implies that $\GAP(\rho)$-typical $\psi$ have entanglement entropy $-\tr(\rho_a^\psi \log \rho_a^\psi)\approx -\tr(\rho_a\log\rho_a)$ with $\rho_a=\tr_b\rho$. 

Since different probability distributions over the unit sphere in a Hilbert space $\Hilbert$ can have the same density matrix, and since the outcome statistics of any experiment depend only on the density matrix, it may seem at first irrelevant to even consider distributions over $\SSS(\Hilbert)$. However, for example, an ensemble of spins prepared so that (about) half are in state $\bigl| \uparrow\bigr\rangle$ and the others in $\bigl| \downarrow\bigr\rangle$ is physically different from a uniform ensemble over $\SSS(\CCC^2)$, even though both ensembles have density matrix $\tfrac{1}{2}I$. Likewise, for an ensemble of particles prepared by taking them from a system in thermal equilibrium, the wave function is GAP-distributed (see Section~\ref{sec:background}). More basically, probability distributions play a key role in any \emph{typicality statement}, i.e., one saying that some condition is satisfied by most wave functions---``most'' relative to a certain distribution; such a statement cannot be formulated in terms of density matrices.

We note that the generalization of canonical typicality from   uniform measures to  GAP measures is not straightforward. First, not \emph{every} measure $\mu$ over $\SSS(\Hilbert_S)$ with a given density matrix $\rho$ with small eigenvalues makes it true that for $\mu$-most $\psi$, $\rho_a^\psi \approx \tr_b \rho$. We give a counter-example in Remark~\ref{rmk:counter} in Section~\ref{sec: main result}. 
Second, if $\rho$ is not close to a multiple of a projection, then GAP$(\rho)$ is far from uniform; specifically, its density will at some points be larger than at others by a factor like $\exp(D)$ (see Remark~\ref{rmk:other-distr}). And third, even measures close to uniform (for example the von Mises-Fisher distribution, see again Remark~\ref{rmk:other-distr}), can fail to satisfy generalized canonical typicality.

In this paper, we prove  generalized canonical typicality in rigorous form by providing error bounds for \eqref{gencan1} at any desired confidence level  that is implicit in the word ``most,'' see  Theorem~\ref{thm: GCT exp bound} and Theorem~\ref{thm:1}.  Compared to the known error bounds based on $u_R$, we can prove more or less the same bounds with $d_R$ replaced by the reciprocal of the largest eigenvalue of $\rho$,
\be
\frac{1}{p_{\max}}:= \frac{1}{\|\rho\|}
\ee
with $\|\cdot\|$ the operator norm. Thus, the approximation is good as soon as no single direction contributes too much to $\rho$. In particular, for $\rho=\rho_R$, our results essentially reproduce the known error bounds. As one central part of our proof, we also establish a variant of L\'evy's lemma \cite{levy1951,MS86,ledoux2001} (a statement about the concentration of measure on a high-dimensional sphere, see below) for GAP measures instead of the uniform measure (Theorem \ref{thm: LM for GAP}). In particular, our version of L\'evy's lemma holds also on infinite dimensional spheres, where the uniform measure does not exist.

Furthermore, we provide several corollaries. The first one shows that for any observable and $\GAP(\rho)$-most $\psi$, the coarse-grained Born distribution is near a $\psi$-independent one (see Remark~\ref{rmk:LevyB} in Section~\ref{sec:statements} for discussion). The second arises from evolving the observable with time and provides a form of  \emph{dynamical typicality} \cite{BG09}, which means that for typical initial wave functions, the time evolution ``looks'' the same; here, ``typical'' refers to the $\GAP(\rho)$ distribution, and ``look'' (which in \cite{TTV22-physik} meant the macroscopic appearance) refers to the Born distribution for the observable considered. In fact, Corollary \ref{cor: dyn typ} even shows that the relevant kind of closeness (to a $t$-dependent but $\psi$-independent distribution) holds jointly for most $t\in [0,T]$. As a further variant (Corollary \ref{cor: dyn can typ}), dynamical typicality also holds when ``look'' refers to $\rho_a^\psi$. Put differently, 
the statement here is that for $\GAP(\rho)$-most $\psi$ and most $t\in[0,T]$,
\be
\rho_a^{\psi_t} \approx \tr_b \rho_t\,,
\ee
where $\psi_t=U_t \,\psi$ and $\rho_t=U_{t} \, \rho \, U_{t}^*$ for an arbitrary unitary time evolution $U_{t}$ (allowing for time-dependent $H_t$). In the original version
of canonical typicality, one particularly considers for $\rho_R$ the micro-canonical density matrix $\rho_\mc$ for a fixed Hamiltonian $H$, for which the time evolution yields nothing interesting because it is invariant anyway; but if we consider arbitrary $\rho$, then $\rho$ can evolve in a non-trivial way even for fixed $H$. 

Another corollary (Corollary \ref{cor:cond}) concerns the conditional wave function $\psi_a$ of $a$ (which is the natural notion of the subsystem wave function for $a$, see Section~\ref{sec:background} for the definition): It is known that if $d_R$ is large, then for $u_R$-most $\psi$ and most bases of $\Hilbert_b$, the Born distribution of $\psi_a$ is approximately $\GAP(\tr_b \rho_R)$. 
We generalize this statement as follows: if $d_b$ is large and $\rho$ has small eigenvalues, then for $\GAP(\rho)$-most $\psi$ and most bases of $\Hilbert_b$, the Born distribution of $\psi_a$ is approximately $\GAP(\tr_b \rho)$. 

The results of this paper can also be regarded as a variant of \emph{equivalence-of-ensembles} in quantum statistical mechanics, i.e., as a new instance of the well-known phenomenon in statistical mechanics that it does not make a big difference whether we use the micro-canonical ensemble or the canonical one (for suitable $\beta$) or another equilibrium ensemble. Indeed, the uniform distribution over the unit sphere in a micro-canonical subspace can be regarded as a quantum analog of the micro-canonical distribution in classical statistical mechanics, and the GAP measure associated with a canonical density matrix as a quantum analog of the canonical distribution; see also Remark \ref{rmk:equiv} in Section \ref{sec:discussion}. 

Our results on generalized canonical typicality \eqref{gencan1} provide two kinds of error bounds based on two strategies of proof. They are roughly analogous to the following two bounds on the probability that a random variable $X$ deviates from its expectation $\EEE X$ by more than $n$ standard deviations $\sqrt{\Var(X)}$: First, the Chebyshev inequality yields the bound $1/n^2$, which is valid for \emph{any} distribution of $X$. Second, the \emph{Gaussian} distribution has very light tails, so if $X$ is Gaussian distributed, then the aforementioned probability is actually smaller than $e^{-n}$ (a type of bound known as a Chernoff bound), so the Chebyshev bound would be very coarse. Likewise, the two kinds of bound we provide are based, respectively, on the Chebyshev inequality and the Chernoff bound (in the form of L\'evy's lemma). The former is polynomial in $p_{\max}$, the latter exponential as in $e^{-1/p_{\max}}$. For the original statement of canonical typicality (using $u_R$), the Chebyshev-type bounds were first given by Sugita \cite{Sugita07}, the Chernoff-type bounds by Popescu et al.\ \cite{PSW05}. Our proof of the Chebyshev-type bounds makes heavy use of results of Reimann~\cite{Reimann08}.

A version of L\'evy's lemma was also established for the mean-value ensemble on a finite-dimensional Hilbert space $\Hilbert$ \cite{MGE}. This is the uniform distribution on $\SSS(\Hilbert)$ restricted to the set $\{\psi\in\SSS(\Hilbert):\langle\psi|A|\psi\rangle =a\}$ for a given observable $A$ and a value $a$ satisfying further conditions. However, as also the authors of \cite{MGE} point out, the physical relevance of this ensemble remains unclear. Also dynamical typicality has been established for the mean-value ensemble, see \cite{RG20} for an overview.

The remainder of this paper is organized as follows: In Section~\ref{sec:motiback}, we elucidate the motivation and background. In Section~\ref{sec: main result}, we formulate and discuss our results. In Section~\ref{sec: proofs}, we provide the proofs. 
In Section~\ref{sec: concl}, we conclude.

\section{Motivation and Background}
\label{sec:motiback}

\subsection{Motivation}
\label{sec:moti}

Canonical typicality is often (rightly) used as a justification and derivation of the canonical density matrix $\rho_\can$  
from something simpler, viz., from the uniform distribution over the unit sphere in an appropriate subspace $\Hilbert_\mc$. So it may appear surprising that here we consider other distributions instead of the uniform one.
That is why we give some elucidation in this section.

The uniform distribution for $\psi$ can appear in either of two roles: as a measure of probability or a measure of typicality. What is the difference? The concept of probability, in the narrower sense used here, refers to a physical situation that occurs many times or can be made to occur many times, so that one can meaningfully speak of the empirical distribution of part of the physical state, such as $\psi$, over the ensemble of trials. In contrast, the concept of typicality, in the sense used here, refers to a hypothetical ensemble and applies also in situations that do not occur repeatedly, such as the universe as a whole, or occur at most a few times; it defines what a typical solution of an equation or theory looks like, or the meaning of ``most.'' Typicality is used in defining what counts as thermal equilibrium (e.g., \cite{Gol01} and references therein), but also in certain laws of nature such as the past hypothesis (a proposed law about the initial micro-state of the universe serving as the basis of the arrow of time; see  \cite[Sec.~5.7]{GLTZ20} for a formulation in terms of typicality). 
Moreover, it plays a key role for the \emph{explanation} of certain phenomena by showing that they occur in ``most'' cases.

The mathematical statements apply regardless of whether we think of the measure as probability or typicality. If we use $u_\mc$ as probability, then the question naturally arises whether the actual distribution of $\psi$ is uniform, and generalizations to other measures are called for. The GAP measures are then particularly relevant, not just as a natural choice of measures, but also because they arise as the thermal equilibrium distribution of wave functions.

But also for $u_\mc$ as a measure of typicality, which is perhaps the more important or more widely used case, the generalization is relevant. The way we practically think of canonical typicality is that if $\psi$ is just ``any old'' wave function of $S$, then $\rho_a^\psi$ will be approximately canonical. But the theorem of original canonical typicality (using $u_\mc$) would require that the coefficients of $\psi$ relative to energy levels of $S$ outside of the micro-canonical energy interval $[E-\Delta E, E]$ are exactly zero, which of course goes against the idea of $\psi$ being ``any old'' $\psi$. Of course, we would expect that the canonicality of $\rho_a^\psi$ does not depend much on whether other coefficients are exactly zero or not. And the theorems in this paper show that this is correct! They show that if the $\rho$ we start from is not $\rho_\mc$, then the crucial part of the reasoning (the typical-$\psi$ part) still goes through, just with corrections reflected in the deviation of $\tr_b \rho$ from $\tr_b \rho_\mc$ (which, by the way, will be minor for $\rho=\rho_\can$ with appropriate inverse temperature $\beta$). More generally, the theorems in this paper prove the robustness of canonical typicality towards changes in the underlying measure.

The results of this paper also show that when computing the typical reduced state $\rho_a^\psi$ for ``any old'' $\psi$, we can start from various choices of $\rho$ of the whole, as long as they yield approximately the same $\tr_b \rho$. The results thus provide researchers with a new angle of looking at canonical typicality: it is OK to imagine ``any old'' $\psi$, and not crucial to start from $u_\mc$.

More generally, our results are a kind of equivalence-of-ensembles statement in the quantum case, and thus add to the picture consisting of various senses in which different thermal equilibrium ensembles are practically equivalent, in this case with ``ensemble'' meaning ensemble of wave functions (i.e., measures over the unit sphere). Again, it plays a role that the GAP measures arise as the thermal equilibrium distribution of wave functions, and thus as an analog of the canonical ensemble in classical statistical mechanics. 
This means also that if $\psi$ is itself a conditional wave function, a case in which we know \cite{GLTZ06b,GLMTZ15} that (for high dimension and most orthonormal bases) $\psi$ is approximately GAP distributed, then canonical typicality applies. 
A special application concerns the thermodynamic limit, for which it is desirable to think of the conditional wave function $\psi_A$ of a region $A$ in 3-space as obtained from $\psi_{A'}$ for a larger $A'\supset A$, which in turn is obtained from $\psi_{A''}$ for an even larger $A'' \supset A'$, and so on. Then for each step, $\psi_{A'}$ (etc.) is GAP-distributed.

By the way, the results here also have the converse implication of supporting the naturalness of the GAP measures. One might even consider a version of the past hypothesis that uses, as the measure of typicality, a GAP measure instead of the uniform distribution over the unit sphere in some subspace of the Hilbert space of the universe.

\subsection{Mathematical Setup and Some Background}
\label{sec:background}

One often considers the uniform distribution over the unit sphere in a subspace $\Hilbert'$ of a system's Hilbert space $\Hilbert$. While this distribution is associated with a density matrix given by the normalized projection to $\Hilbert'$, the measure $\GAP(\rho)$ forms an analog of it for an \emph{arbitrary} density matrix. We now give its definition and that of some other mathematical concepts we use. 

Throughout this paper, all Hilbert spaces $\Hilbert$ are assumed to be separable, i.e., to have either a finite or a countably infinite orthonormal basis (ONB). The unit sphere $\SSS(\Hilbert)$ is always equipped with the Borel $\sigma$-algebra.

\paragraph{Density matrix.}
To any probability measure $\mu$ on $\mathbb{S}(\Hilbert)$ we can associate a density matrix $\rho_\mu$ by
\begin{align}
    \rho_\mu := \int_{\mathbb{S}(\Hilbert)} \mu(d\psi) |\psi\rangle \langle\psi|
\end{align}
(which always exists \cite[Lemma 1]{Tum20}). Note that if $\mu$ has mean zero then $\rho_\mu$ is the covariance matrix of $\mu$. It will turn out for $\mu = \GAP(\rho)$ that $\rho_{\mu} = \rho$. 

\paragraph{GAP measure.}
The measure $\GAP(\rho)$ was first introduced for finite-dimensional $\Hilbert$ by Jozsa, Robb, and Wootters \cite{JRW94}, who named it \textit{Scrooge measure}.\footnote{Named after Ebenezer Scrooge, a fictional character in and the protagonist of Charles Dickens' novella \textit{A Christmas Carol} (1843) who is known for being very stingy. As Jozsa et al. argue, the gap measure is in some sense the most spread-out distribution on $\mathbb{S}(\Hilbert)$ with density matrix $\rho$ and they choose the name ``Scrooge measure'' because the measure is ``particularly stingy with its information.''} 
Among several equivalent definitions \cite{GLMTZ15}, we use the following one based on Gaussian measures. Let $\Hilbert$ be separable and $\rho$ a density matrix on $\Hilbert$ with eigenvalues $p_n$ and eigen-ONB $(|n\rangle)_{n=1\ldots\dim\Hilbert}$, i.e.,
\begin{align}
    \rho = \sum_n p_n |n\rangle\langle n|.
\end{align}
A complex-valued random variable $Z$ will be said to be Gaussian with mean $z\in\CCC$ and variance $\sigma^2>0$ if and only if $\mathrm{Re}\,Z$ and $\mathrm{Im}\,Z$ are independent real Gaussian random variables with mean $\mathrm{Re}\,z$ respectively $\mathrm{Im}\,z$ and each with variance $\sigma^2/2$. Let $(Z_n)_{n=1\ldots\dim\Hilbert}$ be a sequence of independent $\mathbb{C}$-valued Gaussian random variables with mean 0 and variances
\begin{align}
    \mathbb{E}|Z_n|^2 = p_n.
\end{align}
Then we define $\mathrm{G}(\rho)$ to be the distribution of the random vector 
\begin{align}
    \Psi^{\mathrm{G}} := \sum_n Z_n |n\rangle,
\end{align}
i.e., the Gaussian measure on $\Hilbert$ with mean 0 and covariance operator $\rho$. (It is known \cite{Pro56} in general that for every $\phi\in\Hilbert$ and every positive trace-class operator $\rho$ there exists a unique Gaussian measure on $\Hilbert$ with mean $\phi$ and covariance operator $\rho$.) 
Note that
\begin{align}\label{EGrho1}
    \mathbb{E}\|\Psi^{\mathrm{G}}\|^2 = \sum_n \mathbb{E}|Z_n|^2 = \sum_n p_n =1,
\end{align}
which also shows that $\|\Psi^\mathrm{G}\|<\infty$ almost surely,
but in general $\|\Psi^{\mathrm{G}}\|\neq 1$, i.e., $\mathrm{G}(\rho)$ is not a distribution on the sphere $\mathbb{S}(\Hilbert)$. Projecting the measure $\mathrm{G}(\rho)$ to the sphere $\mathbb{S}(\Hilbert)$ would not result in a measure with density matrix $\rho$; therefore we first  adjust the density of $\mathrm{G}(\rho)$ and define the \textit{adjusted Gaussian measure} $\GA(\rho)$ on $\Hilbert$   as the measure that has density $\|\psi\|^2$ relative to $\mathrm{G}(\rho)$, i.e.,
\begin{align}
    \GA(\rho)(d\psi) := \|\psi\|^2 \, \mathrm{G}(\rho)(d\psi),
\end{align}
which is a probability measure by virtue of \eqref{EGrho1}.
It will turn out below that  $\|\psi\|^2$ is the right factor to ensure that $\rho_{\GAP(\rho)}=\rho$.

Let $\Psi^{\GA}$ be a $\GA(\rho)$-distributed random vector. We define $\GAP(\rho)$ to be the distribution of
\begin{align}
    \Psi^{\GAP} := \frac{\Psi^{\GA}}{\|\Psi^{\GA}\|}.
\end{align}
Note that the denominator is almost surely non-zero (because every 1-element subset of $\Hilbert$ has $\mathrm{G}(\rho)$-measure 0 because every $Z_n$ has continuous distribution).
With this we find that indeed
\begin{subequations}
\begin{align}
    \rho_{\GAP(\rho)} &= \int_{\mathbb{S}(\Hilbert)} \GAP(\rho)(d\psi) \,|\psi\rangle\langle\psi|\\ 
    &= \int_{\Hilbert} \GA(\rho)(d\psi) \frac{1}{\|\psi\|^2} |\psi\rangle\langle\psi| \\
    &= \int_{\Hilbert} \mathrm{G}(\rho)(d\psi) ~ |\psi\rangle\langle\psi| = \rho.
\end{align}
\end{subequations}
See \cite{Tum20} for a complete proof of existence and uniqueness of $\GAP(\rho)$ for every density matrix $\rho$.

$\GAP(\rho)$ can also be characterized as the minimizer of the ``accessible information'' functional under the constraint that its density matrix is $\rho$ \cite{JRW94}. If all eigenvalues of $\rho$ are positive and $D:=\dim\Hilbert<\infty$, then $\GAP(\rho)$ possesses a density relative to the uniform distribution $u$ on $\mathbb{S}(\Hilbert)$ \cite{GLTZ06b,GLMTZ15},
\begin{align}\label{GAPdensity}
    \GAP(\rho)(d\psi) = \frac{D}{\det \rho}\langle\psi|\rho^{-1}|\psi\rangle^{-D-1}\, u(d\psi)\,.
\end{align}

It was argued in \cite{GLTZ06b} and mathematically justified in \cite{GLMTZ15} that GAP measures describe the thermal equilibrium distribution of the (conditional) wave function of the system if $\rho$ is a canonical density matrix. 

It was also shown in \cite{GLTZ06b} that $\GAP$ is equivariant under unitary transformations, i.e., for all density matrices $\rho$, all unitary operators  $U$ on  $\Hilbert$, and all measurable sets $M\subset \SSS(\Hilbert)$ one has
\begin{equation}\label{equivariant}
\GAP(U\rho U^*) (M) =  \GAP(\rho)(UM)\,.
\end{equation}
In particular, $\GAP$ is equivariant under   unitary time evolution, and,
as a consequence, $\GAP(\rho_t)$ is the relevant distribution on $\SSS(\Hilbert)$ whenever the system starts in thermal equilibrium with respect to some Hamiltonian $H_0$ and evolves according to any Hamiltonian $H_t$ at later times.
More generally, the results of \cite{GLMTZ15} (and their extension in Corollary~\ref{cor:cond}) show that if a system has density matrix $\rho$ arising from entanglement, then its (conditional) wave function (relative to a typical basis, see below) is asymptotically GAP-distributed. Thus, GAP is the correct distribution in many practically relevant cases. On top of that, when we have no further restriction than that the density matrix is $\rho$, then the natural concept of a ``typical $\psi$'' should refer to the most spread-out distribution compatible with $\rho$, which is $\GAP(\rho)$. 

Finally, let us remark that $\GAP(\rho)$ is also invariant under global phase changes, i.e., $\GAP(\rho)(M) =\GAP(\rho)(e^{i\varphi }M)$ for all measurable $M\subset\SSS(\Hilbert)$ and  $\varphi\in\RRR$. Hence, $\GAP(\rho)$ naturally also defines a probability distribution on the projective space of complex rays in $\Hilbert$ and all results presented in the following can be equivalently formulated for rays instead of vectors.

\begin{rmk}\label{rmk:EEEtrb}
In terms of $\rho_\mu$, we can easily formulate and prove a weaker version of our main result \eqref{gencan1}; this version is related to \eqref{gencan1} in more or less the same way as the statement that in a certain population, the \emph{average} height is 170 cm, is related to the stronger statement that in that population, \emph{most} people are 170 cm tall. The weaker version asserts that the \emph{average} of $\rho^\psi_a$ over $\psi$ using the $\GAP(\rho)$ distribution is equal to $\tr_b \rho$, whereas the statement about \eqref{gencan1} was that \emph{most} $\psi$ relative to $\GAP(\rho)$ have $\rho^\psi_a$ (approximately) equal to $\tr_b \rho$. On the other hand, the statement about the average is stronger because it asserts, not \emph{approximate} equality, but \emph{exact} equality. On top of that, the average statement is not limited to the GAP measure but holds for \emph{any} probability measure $\mu$. Here is the full statement: \emph{for separable $\Hilbert=\Hilbert_a \otimes \Hilbert_b$, any probability measure $\mu$ on $\SSS(\Hilbert)$, and a random vector $\psi$ with distribution $\mu$,} 
\be\label{mu-avg-rhopsi}
\EEE_\mu \rho_a^\psi = \tr_b \rho_\mu \,.
\ee
Indeed, $\tr_b$ commutes with $\mu$-integration,\footnote{Since we could not find a good reference for this fact, we have included a proof in Section~\ref{sec:EEEtrb}.} so
\begin{subequations}
\begin{align}
    \EEE_\mu \rho_a^\psi 
    &= \int_{\SSS(\Hilbert)} \mu(d\psi) \, \tr_b |\psi\rangle\langle\psi|\\
    &= \tr_b\int_{\SSS(\Hilbert)} \mu(d\psi) \, |\psi\rangle\langle\psi| \\[2mm]
    &= \tr_b \rho_\mu \,.
\end{align}
\end{subequations}
\hfill$\diamond$
\end{rmk}

\paragraph{Norms.}
The distance between   two density matrices will be measured in the \textit{trace norm}
\begin{align}
    \|M\|_{\tr} := \tr |M| = \tr\sqrt{M^*M},
\end{align}
where $M^*$ denotes the adjoint operator of $M$. 
If $M$ can be diagonalized through an orthonormal basis (ONB), then $\|M\|_{\tr}$ is the sum of the absolute eigenvalues. We will also sometimes use the \emph{operator norm}
\be
\|M\| := \sup_{\|\psi\|=1} \|M\psi\|\,,
\ee
which, if $M$ can be diagonalized through an ONB, is the largest absolute eigenvalue.

\paragraph{Purity.}
For a density matrix $\rho$, its \emph{purity} is defined as $\tr \rho^2$. In terms of the spectral decomposition $\rho=\sum_n p_n |n\rangle\langle n|$, the purity is $\tr \rho^2=\sum_n p_n^2$, which can be thought of as the average size of $p_n$. In particular, the purity is positive and $\leq 1$; it is $=1$ if and only if $\rho$ is pure, i.e., a 1d projection; for a normalized projection $\rho_R=P_R/d_R$, the purity is $1/d_R$; conversely, $1/$purity can be thought of as the effective number of dimensions over which $\rho$ is spread out. It also easily follows that
\be
\tr \rho^2 \leq \|\rho\| \leq \sqrt{\tr\rho^2} \leq \sqrt{\|\rho\|}
\ee
because $p_n^2 \leq p_n \|\rho\|$, and if $p_{n_0}$ is the largest eigenvalue, then $p_{n_0}^2 \leq \sum_n p_n^2$ because all other terms are $\geq 0$. In words, the average $p_n$ is no greater than the maximal $p_n$, which is bounded by the square root of the average $p_n$ (and the square root of the maximal $p_n$).

\paragraph{Conditional wave function.} For $\Hilbert=\Hilbert_a\otimes\Hilbert_b$, an ONB $B=(|m\rangle_b)_{m=1\ldots\dim\Hilbert_b}$ of $\Hilbert_b$, and $\psi\in\SSS(\Hilbert)$, the conditional wave function $\psi_a$ \cite{DGZ92,GN99,GLTZ06b} of system $a$ is a random vector in $\Hilbert_a$ that can be constructed by choosing a random one of the basis vectors $|m\rangle_b$, let us call it $|M\rangle_b$, with the Born distribution 
\be\label{Born1}
\PPP(M=m) = \bigl\| {}_b\scp{m}{\psi}\bigr\|^2_a \,,
\ee
taking the partial inner product of $|M\rangle_b$ and $\psi$, and normalizing:
\be\label{psiadef}
\psi_a := \frac{{}_b\scp{M}{\psi}}{\|{}_b\scp{M}{\psi}\|_a}\,.
\ee
(Note that the event that $\|{}_b\scp{M}{\psi}\|_a=0$ has probability 0 by \eqref{Born1}. 
In the context of Bohmian mechanics, the expression ``conditional wave function'' refers to the position basis and the Bohmian configuration of $b$ \cite{DGZ92}; but for our purposes, we can leave it general.)

We can also think of $\psi_a$ as arising from $\psi$ through a quantum measurement with eigenbasis $B$ on system $b$, which leads to the collapsed quantum state $\psi_a \otimes |M\rangle_b$. Correspondingly, we call the  distribution of $\psi_a$ in $\SSS(\Hilbert_a)$ the \emph{Born distribution of $\psi_a$} and denote it by $\mathrm{Born}_a^{\psi,B}$. However, when considering $\psi_a$, we will not assume that any observer actually, physically carries out such a quantum measurement; rather, we use $\psi_a$ as a theoretical concept of a
wave function associated with the subsystem $a$. It is related to the reduced density matrix $\rho^\psi_a$ in a way similar to how a conditional probability distribution is to a marginal distribution,
\be
\EEE |\psi_a\rangle \langle \psi_a| = \rho_a^\psi \,.
\ee
$\psi_a$ is also related to the GAP measure, in fact in two ways. First, when we average $\mathrm{Born}^{\psi,B}_a$ over all ONBs $B$ (using the uniform distribution corresponding to the Haar measure), then we obtain $\GAP(\rho_a^\psi)$ \cite[Lemma 1]{GLMTZ15}. Put differently, if we think of both $M$ and $B$ as random and $\psi_a$ thus as doubly random, then its (marginal) distribution is $\GAP(\rho_a^\psi)$;  
put more briefly, $\GAP(\rho_a^\psi)$ is the distribution of the collapsed pure state in $a$ after a purely random quantum measurement in $b$ on $\psi$. Second, if $d_b$ is large, then even conditionally on a single given $B$, the distribution of $\psi_a$ is close to a GAP measure for most $B$ and most $\psi$ according to a GAP measure on $\Hilbert_a\otimes \Hilbert_b$; this is the content of Corollary~\ref{cor:cond} below. 

\section{Main Results} 
\label{sec: main result}

In this section, we present and discuss our main results about generalized canonical typicality. In the following, we use the notation $\mu(f)$ for the average of the function $f$ under the measure $\mu$,
\be\label{mufdef}
\mu(f) := \int \mu(d\psi) \, f(\psi) \,.
\ee
Note that, by \eqref{mu-avg-rhopsi},
\be\label{GAPrhopsi}
\GAP(\rho)(\rho_a^\psi)=\tr_b \rho\,.
\ee
The statement of our generalized canonical typicality  differs in that it concerns \emph{approximate} equality and holds for the \emph{individual} $\rho_a^\psi$, not only for its average.

\subsection{Statements}
\label{sec:statements}

We first formulate our main theorem on canonical typicality for GAP measures and the underlying variant of L\'evy's lemma for GAP measures. We then give a list of further  consequences  of this generalized version of L\'evy's lemma, including  results on dynamical typicality and the fact that the typical Born distribution of conditional wave functions  is itself a GAP measure. 
At the end of this section we also state a slightly weaker version of our main theorem that is not based on L\'evy's lemma but instead allows for a rather elementary proof based on the Chebyshev inequality. Finally, the known bounds for uniformly distributed $\psi$ will be stated in Remark~\ref{rmk:comparison} in Section~\ref{sec:discussion} for comparison. 

\begin{thm}[Generalized canonical typicality, exponential bounds]\label{thm: GCT exp bound}
     Let $\Hilbert_a$ and $\Hilbert_b$ be Hilbert spaces with $\Hilbert_a$ having finite dimension $d_a$ and $\Hilbert_b$ being separable, 
    and let $\rho$ be a density matrix on $\Hilbert = \Hilbert_a \otimes \Hilbert_b$. Then for every $\delta>0$,
    \begin{align}\label{expbounddelta}
        \GAP(\rho)\Biggl\{\psi\in\mathbb{S}(\Hilbert): \bigl\|\rho_a^\psi-\tr_b\rho\bigr\|_{\tr}\leq c d_a \sqrt{\ln\left(\frac{12d_a^2}{\delta}\right)\|\rho\|}\Biggr\} \geq 1-\delta,
    \end{align}
    where $c=48\pi$.
\end{thm}

\begin{rmk}\label{rmk: exp bound}
The relation \eqref{expbounddelta} can equivalently be formulated as a bound on the confidence level, given the allowed deviation: {\it For every $\varepsilon\geq 0$,}
\begin{align}\label{expboundepsilon}
    \GAP(\rho)\Bigl\{\psi\in\mathbb{S}(\Hilbert): \bigl\|\rho_a^\psi-\tr_b\rho\bigr\|_{\tr} > \varepsilon\Bigr\} \leq 12d_a^2 \exp\left(-\frac{\tilde{C}\varepsilon^2}{d_a^2 \|\rho\|}\right),
\end{align}
{\it where} $\tilde{C}=\frac{1}{2304\pi^2}$. This form makes it visible why we call Theorem~\ref{thm: GCT exp bound} an ``exponential bound'': because the bound on the probability of too large a deviation is exponentially small in $1/\|\rho\|$. In contrast, the bound \eqref{polyboundepsilon} is polynomially small in $\tr\rho^2$.\hfill$\diamond$
\end{rmk}

A key tool for proving Theorem~\ref{thm: GCT exp bound} is Theorem~\ref{thm: LM for GAP} below, a variant of L\'evy's lemma for GAP measures. Recall the notation \eqref{mufdef}.

\begin{thm}[L\'evy's lemma for GAP measures]\label{thm: LM for GAP}
  Let $\Hilbert$ be a separable Hilbert space, let $f:\mathbb{S}(\Hilbert)\to\mathbb{R}$ be a Lipschitz continuous function with Lipschitz constant\footnote{A Lipschitz constant refers to a metric on the domain, and two metrics are often considered on the sphere: the spherical metric (distance along the sphere, $d_\mathrm{sph}(\psi,\phi)=\arccos \mathrm{Re}\scp{\psi}{\phi}$) and the Euclidean metric (distance in the ambient space across the interior of the sphere, $d_\mathrm{Eucl}(\psi,\phi)=\|\psi-\phi\|$). We use the spherical metric, as did \cite{MS86,PSW05,PSW06}, but since $d_\mathrm{Eucl}(\psi,\phi)\leq d_\mathrm{sph}(\psi,\phi)\leq \tfrac{\pi}{2} \, d_\mathrm{Eucl}(\psi,\phi)$, using the Euclidean metric would at most change the Lipschitz constants by a factor of $\tfrac{\pi}{2}$.} $\eta$, let $\rho$ be a density matrix on $\Hilbert$, and let $\varepsilon\geq 0$. 
  Then
  \be\label{LMGAP}
      \GAP(\rho)\Bigl\{ \psi\in\mathbb{S}(\Hilbert): \bigl|f(\psi)-\GAP(\rho)(f) \bigr|>\varepsilon\Bigr\} \leq 6 \exp\left(-\frac{C\varepsilon^2}{\eta^2 \|\rho\|} \right),
  \ee
  where $C=\frac{1}{288\pi^2}$.
\end{thm}

\begin{rmk}\label{rmk:LevyfC}
    The statement remains true for \emph{complex}-valued $f$ if we replace the constant factor 6 in \eqref{LMGAP} by 12 and $C$ by $C/2$, as follows from considering the real and imaginary parts of $f$ separately.\hfill$\diamond$
\end{rmk}

As an immediate consequence of Theorem~\ref{thm: LM for GAP} for  $f(\psi) = \langle\psi|B|\psi\rangle$, which has Lipschitz constant $\eta \leq 2\|B\|$ \cite[Lemma 5]{PSW05}, we obtain:

\begin{cor}\label{cor:LevyB}
    Let $\rho$ be a density matrix and $B$ a bounded operator on the separable Hilbert space $\Hilbert$. For every $\varepsilon\geq 0$,
    \be\label{LevyB}
    \GAP(\rho)\Bigl\{ \psi\in\mathbb{S}(\Hilbert): \bigl|\langle\psi|B|\psi\rangle-\tr(\rho B) \bigr|>\varepsilon\Bigr\} \leq 12 \exp\left(-\frac{\tilde{C}\varepsilon^2}{\|B\|^2 \|\rho\|} \right)
    \ee
    with $\tilde{C} = \frac{1}{2304\pi^2}$.
\end{cor}

\begin{rmk}\label{rmk:LevyB} 
Corollary~\ref{cor:LevyB} provides an extension to GAP measures of the known fact \cite{Reimann07} that $\scp{\psi}{B|\psi}$ has nearly the same value for most $\psi$ relative to the uniform distribution. This kind of near-constancy is different from the near-constancy property of a macroscopic observable, viz., that most of its eigenvalues (counted with multiplicity) in the micro-canonical energy shell are nearly equal. Here, in contrast, nothing (except boundedness) is assumed about the distribution of eigenvalues of $B$. In particular, if $B$ is a self-adjoint observable, then a typical $\psi$ may well define a non-trivial probability distribution over the spectrum of $B$, not necessarily a sharply peaked one. The near-constancy property asserted here is that the \emph{average} of this probability distribution is the same for most $\psi$. In fact, it also follows that the \emph{probability distribution itself} is the same for most $\psi$ (``distribution typicality''), at least on a coarse-grained level (by covering the spectrum of $B$ with not-too-many intervals) and provided that many dimensions participate in $\rho$. This follows from inserting spectral projections of the observable for $B$ in \eqref{LevyB}.\hfill$\diamond$ 
\end{rmk}

In contrast to the uniform distribution on the sphere in the micro-canonical subspace, which is invariant under the unitary time evolution, $\GAP(\rho_0)$ will in general evolve, in fact to $\GAP(\rho_t)$ by \eqref{equivariant}. This leads to questions about what the history $t\mapsto \psi_t$ looks like. Inserting $U_t^*BU_t$ for $B$ in \eqref{LevyB} leads us to the first equation in the following variant of ``dynamical typicality'' for GAP measures.
\begin{cor}\label{cor: dyn typ}
    Let $\Hilbert$ be a separable Hilbert space, $B$ a bounded operator and $\rho$  a density matrix on $\Hilbert$, and $t\mapsto U_{t}$ a  measurable family of unitary operators. Then for every $\varepsilon,t\geq 0$,
    \begin{align}
        \GAP(\rho)\Bigl\{\psi\in\mathbb{S}(\Hilbert):\left|\langle\psi_t|B|\psi_t\rangle - \tr(\rho_t B)\right|>\varepsilon\Bigr\} \leq 12 \exp\left(-\frac{\tilde{C}\varepsilon^2}{\|B\|^2 \|\rho\|}\right),\label{cordyntyp1}
    \end{align}
    where $\rho_t = U_{t}\,\rho\, U_{t}^*$, $\psi_t = U_{t}\psi$ and $\tilde{C} = \frac{1}{2304\pi^2}$. Moreover, for every $\varepsilon, T>0$,
    \begin{align}
        \GAP(\rho)\Bigl\{\psi\in\mathbb{S}(\Hilbert): \frac{1}{T}\int_0^T \left|\langle\psi_t|B|\psi_t\rangle - \tr(\rho_t B)\right|\, dt > \varepsilon\Bigr\} \leq 9\exp\left(-\frac{\tilde{C}\varepsilon^2}{36\|B\|^2 \|\rho\|}\right).
    \end{align}
\end{cor}

Clearly, for $U_t$ we have in mind either a unitary group $U_t= \exp(-iHt)$ generated by a time-independent Hamiltonian $H$, or a unitary evolution family $U_{t}$ satisfying $i\frac{d}{dt} U_{t} = H_t U_{t}$ and $U_{0}=I$ generated by a time-dependent Hamiltonian $H_t$. However, the group resp.\ co-cycle structure play no role in the proof. (In \cite{TTV22-physik}, a similar result for the uniform distribution over the sphere in a large subspace was formulated only for time-\emph{independent} Hamiltonians, but the proof given there actually applies equally to time-\emph{dependent} ones.)

The last two corollaries were applications of L\'evy's lemma that did not involve reduced density matrices. We now turn to bi-partite systems again and present two further corollaries. We first ask whether, for $\GAP(\rho_0)$-typical $\psi_0$, the reduced density matrix $\rho_a^{\psi_t}$ remains close to $\tr_b \rho_t$ over a whole time interval $[0,T]$. The following corollary answers this question affirmatively for most times in this interval.

\begin{cor}\label{cor: dyn can typ}
    Let $\Hilbert_a$ and $\Hilbert_b$ be  Hilbert spaces with $\Hilbert_a$ having finite dimension $d_a$ and $\Hilbert_b$ being separable,  $\rho$  a density matrix on $\Hilbert=\Hilbert_a\otimes\Hilbert_b$,  and  $t\mapsto U_{t}$ a  measurable family of unitary operators on $\Hilbert$. Then for every $\varepsilon,T>0$,
    \begin{align}
       \GAP(\rho)\left\{\psi\in\mathbb{S}(\Hilbert):\frac{1}{T}\int_0^T \bigl\|\rho_a^{\psi_t} - \tr_b \rho_t\bigr\|_{\tr}\, dt > \varepsilon\right\} \leq 9 d_a^2 \exp\left(-\frac{\tilde{C}\varepsilon^2}{36 d_a^2 \|\rho\|}\right),
    \end{align}
    where $\rho_t = U_{t}\,\rho\, U_{t}^*$, $\psi_t = U_{t}\psi$ and  $\tilde{C} = \frac{1}{2304\pi^2}$.
\end{cor}

The next corollary expresses that for $\GAP(\rho)$-typical $\psi$, large $d_b$, and small $\tr\rho^2$, the conditional wave function $\psi_a$ (relative to a typical basis) has Born distribution close to $\GAP(\tr_b\rho)$. (Note that we are considering the distribution of $\psi_a$ \emph{conditionally} on a given $\psi$, rather than the \emph{marginal} distribution of $\psi_a$ for random $\psi$, which would be $\int_{\SSS(\Hilbert)} \GAP(\rho)(d\psi) \, \mathrm{Born}^{\psi,B}_a(\cdot)$.) Recall the notation \eqref{mufdef}.

\begin{cor}\label{cor:cond}
    Let $\varepsilon,\delta\in(0,1)$, let $\Hilbert_a$ be a Hilbert space of dimension $d_a\in\NNN$, let $f:\SSS(\Hilbert_a)\to\RRR$ be any  continuous (test) function, and let $\Hilbert_b$ be a Hilbert space of finite dimension $d_b\geq \max\{4,d_a,32\|f\|^2_\infty /\varepsilon^2\delta\}$.     
    Then there is $p>0$ such that for every density matrix $\rho$ on $\Hilbert=\Hilbert_a \otimes \Hilbert_b$ with $\|\rho\|\leq p$,
\begin{multline}\label{typcond}
\GAP(\rho)  \times u_\mathrm{ONB} \Bigl\{(\psi,B) \in \SSS(\Hilbert) \times \mathrm{ONB}(\Hilbert_b): \\
\bigl|\mathrm{Born}_a^{\psi,B}(f) - \GAP(\tr_b\rho)(f)\bigr|< \varepsilon \Bigr\} \geq 1-\delta \,,
\end{multline}
where $\mathrm{Born}_a^{\psi,B}$ is the distribution of the conditional wave function, $\mathrm{ONB}(\Hilbert_b)$ is the set of all orthonormal bases on $\Hilbert_b$, and $u_{\mathrm{ONB}}$ the uniform distribution over this set.
\end{cor}

\begin{rmk}
    We conjecture that the closeness between $\mathrm{Born}_a^{\psi,B}$ and $\GAP(\tr_b\rho)$ is even better than stated in Corollary~\ref{cor:cond}, at least when 0 is not an eigenvalue of $\tr_b \rho$, in the sense that \eqref{typcond} holds not only for continuous $f$ but even for bounded measurable $f$, and in fact uniformly in $f$ with given $\|f\|_\infty$. This conjecture is suggested by using Lemma 6 of \cite{GLMTZ15} instead of Lemma 5, or rather a variant of it with more explicit bounds. 
    \hfill$\diamond$ 
\end{rmk}

Whereas Theorem~\ref{thm: GCT exp bound} is based on the rather technical concentration of measure result Theorem~\ref{thm: LM for GAP}, a slightly weaker statement can be obtained using only the Chebychev inequality and a bound on the variance of random variables of the form $\psi\mapsto \langle\psi|A|\psi\rangle$ with respect to GAP$(\rho)$ given in Proposition~\ref{prop: var} in Section~\ref{sec:PfThm1}. The latter bound is also of   interest in its own right and has already been established for self-adjoint $A$ by Reimann in \cite{Reimann08}.   

\begin{thm}[Generalized canonical typicality, polynomial bounds]\label{thm:1}
Let $\Hilbert_a$ and $\Hilbert_b$ be Hilbert spaces with $\Hilbert_a$ having finite dimension $d_a$ and $\Hilbert_b$ being separable. 
Let $\rho$ be a density matrix on $\Hilbert = \Hilbert_a \otimes \Hilbert_b$ with $\|\rho\|<1/4$. 
Then for every $\delta>0$,
\begin{align}\label{polybounddelta}
    \GAP(\rho)\Biggl\{\psi\in\mathbb{S}(\Hilbert): \bigl\|\rho_a^{\psi} - \tr_b\rho \bigr\|_{\textup{tr}} \leq \sqrt{\frac{28d_a^5 \tr \rho^2}{\delta}} \Biggr\} \geq 1-\delta.
\end{align}
\end{thm}

\begin{rmk}
Again, we can equivalently express Theorem~\ref{thm:1} as a bound on the confidence level $1-\delta$ for any given allowed deviation of $\rho_a^\psi$ from $\tr_b\rho$: {\it For every $\rho$ with $\|\rho\|<1/4$ and every $\varepsilon>0$,}
\be\label{polyboundepsilon}
    \GAP(\rho)\Bigl\{\psi\in\mathbb{S}(\Hilbert): \bigl\|\rho_a^{\psi} - \tr_b\rho \bigr\|_{\textup{tr}} >\varepsilon \Bigr\} \leq \frac{28d_a^5\tr\rho^2}{\varepsilon^2}\,.
\ee
\hfill$\diamond$
\end{rmk}

\begin{rmk} While our main motivation for developing Theorem~\ref{thm:1} is the different strategy of proof, and while the exponential bound of Theorem~\ref{thm: GCT exp bound} will usually be tighter than the polynomial bound of Theorem~\ref{thm:1}, this is not always the case: the bound of Theorem~\ref{thm:1} is actually sometimes better, as the following example shows. Suppose that $\|\rho\|=\frac{1}{\sqrt{D}} = p_1$ and that all other $p_j$ are equal, i.e.,
\begin{align}
    p_j = \frac{1-\frac{1}{\sqrt{D}}}{D-1}
\end{align}
for all $j>1$. Then,
\begin{align}
    \tr\rho^2 = \frac{1}{D}+\frac{1}{D-1}\left(1-\frac{1}{\sqrt{D}}\right)^2 \approx \frac{2}{D},
\end{align}
and for, e.g., $d_a=1000$ and $\varepsilon=0.01$ we find that
\begin{align}
    \frac{28d_a^5}{\varepsilon^2} \frac{2}{D} < 12d_a^2\exp\left(-\frac{\tilde{C}\varepsilon^2\sqrt{D}}{d_a^2}\right)
\end{align}
for $4.67\cdot 10^{13} < D < 9.17\cdot 10^{31}$, i.e., in this example there is a regime in which $D$ is already very large but still the polynomial bound is smaller than the exponential one.
\hfill$\diamond$
\end{rmk}

\subsection{Discussion}
\label{sec:discussion}

\begin{rmk}\label{rmk:size}
{\it System size.}
Theorem~\ref{thm:1} shows, roughly speaking, that as soon as
\begin{align}
    \tr\rho^2 \ll d_a^{-5},\label{ineq: cond small}
\end{align}
$\GAP(\rho)$-most wave functions $\psi$ have $\rho_a^{\psi}$ close to $\tr_b\rho$. If we think of $1/\tr\rho^2$ as the effective number of dimensions participating in $\rho$, and if this number of dimensions is comparable to the full number $D=\dim \Hilbert =d_a d_b$ of dimensions, then 
\eqref{ineq: cond small} reduces to
\begin{align}
    d_a^{5} \ll D.
\end{align}
Since the dimension is exponential in the number of degrees of freedom, this condition roughly means that the subsystem $a$ comprises fewer than $20\%$ of the degrees of freedom of the full system. (The same consideration was carried out in \cite{GHLT15,GHLT16} for the original statement of canonical typicality.) The stronger exponential bound yields that $a$ can even comprise up to $50\%$ of the degrees of freedom \cite{GHLT15,GHLT16}.\hfill$\diamond$
\end{rmk}

\begin{rmk}
{\it Canonical density matrix.}
    A $\rho$ of particular interest is the canonical density matrix
    \be
    \rho_\can = \frac{1}{Z(\beta)} e^{-\beta H}\,.
    \ee
    The relevant condition for generalized canonical typicality to apply to $\rho=\rho_\can$ is that it has small purity $\tr \rho^2$ and small largest eigenvalue $\|\rho\|$. We argue that indeed it does. 
    
    One heuristic reason is equivalence of ensembles: since $\rho_\mc$ has purity $1/d_\mc$ and largest eigenvalue $1/d_\mc$, which is small, the values for $\rho_\can$ should be similarly small. Another heuristic argument is based on the idealization that the system consists of many non-interacting constituents, so that $\Hilbert=\Hilbert_1^{\otimes N}$ and $H=\sum_{j=1}^N I^{\otimes (j-1)}\otimes H_1 \otimes I^{\otimes (N-j)}$, so $\rho_\can =\rho_{1\can}^{\otimes N}$. It is a general fact that for tensor products $\rho_1\otimes \rho_2$ of density matrices, the purities multiply, $\tr (\rho_1\otimes \rho_2)^2=(\tr \rho_1^2)(\tr \rho_2^2)$, and the largest eigenvalues multiply, $\|\rho_1\otimes \rho_2\|=\|\rho_1\| \, \|\rho_2\|$. Thus, the purity of $\rho_\can$ is the $N$-th power of that of $\rho_{1\can}$, and likewise the largest eigenvalue. Since $N\gg 1$ and the values of $\rho_{1\can}$ are somewhere between 0 and 1, and not particularly close to 1, the values of $\rho_\can$ are close to 0, as claimed. We expect that mild interaction does not change that picture very much.\hfill$\diamond$
\end{rmk}

\begin{rmk}\label{rmk:classical}
{\it Classical vs.\ quantum.}
    While classically, a typical phase point from a canonical ensemble is also a typical phase point from some micro-canonical ensemble, a typical wave function from $\GAP(\rho_\beta)$ does not lie in any micro-canonical subspace $\Hilbert_{\mathrm{mc}}$ (if $\Hilbert \neq \Hilbert_{\mathrm{mc}}$) and even if it does lie in an $\Hilbert_{\mathrm{mc}}$, then it is not typical from that subspace; that is because typical wave functions are superpositions of many energy eigenstates, and the weights of these eigenstates in $\rho_{\mathrm{mc}}$ and $\rho_\can$ are reflected in the weights of these eigenstates in the superposition. Therefore, already in the case that $\rho$ is a canonical density matrix, Theorems \ref{thm:1} and \ref{thm: GCT exp bound} are not just simple consequences of canonical typicality but independent results.\hfill$\diamond$
\end{rmk}

\begin{rmk} {\it Equivalence of ensembles.}\label{rmk:equiv}
    We can now state more precisely the sense in which our results provide a version of equivalence of ensembles. It is well known that if $a$ and $b$ interact weakly and $b$ is large enough, then both $\rho_{\mathrm{mc}}$ and $\rho_\can$ in $\Hilbert_S=\Hilbert_a\otimes \Hilbert_b$ lead to reduced density matrices close to the canonical density matrix \eqref{rhoacan} for $a$, $\tr_b\rho_{\mathrm{mc}}\approx \rho_{a,\can}\approx \tr_b \rho_\can$, provided the parameter $\beta$ of $\rho_\can$ and $\rho_{a,\can}$ is suitable for the energy $E$ of $\rho_{\mathrm{mc}}$. Hence, Theorems~\ref{thm:1} and \ref{thm: GCT exp bound} yield that we can start from either $u_{\mathrm{mc}}$ or $\GAP(\rho_\can)$ and obtain for both ensembles of $\psi$ that $\rho_a^\psi$ is nearly constant and nearly canonical.\hfill$\diamond$
\end{rmk}

\begin{rmk}\label{rmk:comparison}
    {\it Comparison to original theorems.} The original, known theorems about canonical typicality, which refer to the uniform distribution over a suitable sphere instead of a GAP measure, are still contained in our theorems as special cases, except for worse constants and in some places additional factors of $d_a$ (which we usually think of as constant as well).
    For more detail, let us begin with the known theorem analogous to Theorem~\ref{thm:1} (formulated this way in \cite[Eq.~(32)]{GHLT16}, based on arguments from \cite{Sugita07}):

\begin{thm}[Canonical typicality, polynomial bounds]\label{thm:4}
Let $\Hilbert_a$ and $\Hilbert_b$ be Hilbert spaces of respective dimensions $d_a, d_b \in \mathbb{N}$, $\Hilbert = \Hilbert_a \otimes \Hilbert_b$, $\Hilbert_R$ be any subspace of $\Hilbert$ of dimension $d_R$, $\rho_R$ be $1/d_R$ times the projection to $\Hilbert_R$, and $u_R$ the uniform distribution over $\SSS(\Hilbert_R)$. Then for every $\delta>0$,
\begin{align}\label{upolybounddelta}
    u_R\Biggl\{\psi\in\mathbb{S}(\Hilbert_R): \bigl\|\rho_a^{\psi} - \tr_b\rho_R \bigr\|_{\textup{tr}} \leq \frac{d_a^2}{\sqrt{\delta d_R}} \Biggr\} \geq 1-\delta.
\end{align}
\end{thm}

    When we apply our Theorem~\ref{thm:1} to $\rho=\rho_R$ (and assume $d_R\geq 4$), we obtain that $\GAP(\rho)=u_R$, $\tr\rho^2=1/d_R$, and almost exactly the bound \eqref{upolybounddelta} except for a (rather irrelevant) factor $\sqrt{28}$ and $d_a^{2.5}$ instead of $d_a^2$. Further explanation of how this different exponent comes about can be found in Section~\ref{sec: rmk8 explanations}.

\begin{thm}[Canonical typicality, exponential bounds \cite{PSW05,PSW06}] \label{thm:PSW} With the notation and hypotheses as in Theorem~\ref{thm:4}, for every $\delta>0$ such that
\begin{align}
    \delta < 4\exp\left(-d_a^2/(18\pi^3)\right),
\end{align}
\begin{equation}\label{cantyp} u_R \Biggl\{ \psi \in \SSS(\Hilbert_R): \bigl\|\rho_a^\psi -
  \tr_b \rho_R  \bigr\|_{\tr} \leq  2\sqrt{\frac{18\pi^3}{d_R}\ln(4/\delta)}
  \Biggr\} \geq 1-\delta\,.
\end{equation}
\end{thm}
This theorem was stated slightly differently in \cite{PSW05,PSW06}; we give the derivation of this form in Section~\ref{sec: rmk8 explanations}.
Again, the bound agrees with the one \eqref{expbounddelta} provided by Theorem~\ref{thm: GCT exp bound} for $\rho=\rho_R$ (so $\|\rho\|=1/d_R$) up to worse constants and additional factors of $d_a$. 

    Next, here is the standard statement of L\'evy's lemma:\footnote{L\'evy's original 1922 statement (reprinted as a second edition in \cite[Sec.~3.I.9]{levy1951}) was that if a hypersurface $S\subset\SSS(\RRR^d)$ divides the sphere in two regions of equal area then its $\varepsilon$-neighborhood has area greater than or equal to that of the $\varepsilon$-neighborhood of an equator, which in turn \cite[Sec.~3.I.6]{levy1951} has nearly full area if the dimension $d$ is large enough. As pointed out by, e.g., Milman and Schechtman \cite{MS86}, it follows for a function $f:\SSS(\RRR^d)\to\RRR$ with Lipschitz constant $\eta$ (by taking $S=f^{-1}(m)$ and $m$ the median of $f$) that most points $\psi$ have $f(\psi)$ close to $m$ if $d$ is large enough. The variant quoted here referring to the mean instead of the median is due to Maurey and Pisier \cite{Pis} and also described in \cite[App.~V]{MS86}.}
    
\begin{thm}[L\'evy's Lemma \cite{MS86}] \label{thm: LM}
  Let $\Hilbert$ be a Hilbert space of finite dimension $D:=\dim \Hilbert \in \mathbb{N}$, let $f:\mathbb{S}(\Hilbert)\to\mathbb{R}$ be a function with Lipschitz constant $\eta$, let $u$ be the uniform distribution over $\SSS(\Hilbert)$, and let $\varepsilon>0$.
  Then
\be\label{Levy}
u\Bigl\{\psi\in\SSS(\Hilbert): \bigl|f(\psi)-u(f)\bigr| > \varepsilon \Bigr\} \leq
4\exp\left(-\frac{\hat{C} D\varepsilon^2}{\eta^2}\right)\,,
\ee
where $\hat{C}=\frac{2}{9\pi^3}$.
\end{thm}

When we apply our Theorem~\ref{thm: LM for GAP} to $\rho=I/D$, we obtain that $\GAP(\rho)=u$, $\|\rho\|=1/D$, and exactly the bound \eqref{Levy} except for worse constants. Note that Theorem~\ref{thm: LM for GAP} holds also for infinite dimensional separable $\Hilbert$.

We turn to previous results for dynamical typicality. In \cite{TTV22-physik}, an inequality analogous to the bound \eqref{cordyntyp1} of Corollary~\ref{cor: dyn typ} was proven for the uniform distribution over the sphere in a subspace. In \cite{MGE}, variants of L\'evy's lemma and dynamical typicality were established for the mean-value ensemble of an observable $A$ for a value $a\in\RRR$, defined by restricting the uniform distribution on $\SSS(\CCC^D)$ to the set $\{\psi\in\SSS(\CCC^D):\langle \psi| A|\psi\rangle =a\}$ and normalizing afterwards. 
However, the physical relevance of this ensemble is unclear, since, in general, the mean value of an observable is itself no observable, and thus it is unclear how this ensemble could be  prepared or occur in an experiment.\hfill$\diamond$
\end{rmk}

\begin{rmk}\label{rmk:other-distr}
{\it L\'evy's lemma for other distributions.} L\'evy's lemma, although it applies to the uniform and GAP measures, does not apply to \emph{all} rather-spread-out distributions on the sphere; it is thus a non-trivial property of the family of GAP measures. 

This can be illustrated by means of the von Mises-Fisher (VMF) distribution, a well known and natural probability distribution on the unit sphere $\SSS(\RRR^D)$ in $\RRR^D$ that is different from the GAP measure. It has parameters $\kappa\in\mathbb{R}_+$ and $\mu\in \SSS(\RRR^D)$ and can be obtained from a Gaussian distribution in $\RRR^D$ with mean $\mu$ and covariance $\kappa^{-1}I$ by conditioning on $\SSS(\RRR^D)$. The analog of L\'evy's lemma for the von Mises-Fisher distribution is false; this can be seen as follows. Its density
\be
g(x)= C(D,\kappa) \,\exp\bigl( \kappa \,\langle \mu,x\rangle_{\mathbb{R}^D} \bigr)
\ee
with respect to the uniform distribution $u$ on $\SSS(\RRR^D)$ varies at most by a factor of $e^{2\kappa}$ when varying $x$ (while keeping $D$ and $\kappa$ fixed). For a given Lipschitz function $F$ on the sphere, insertion of $F(x)\,g(x)$ for $f(x)$ in a real variant of L\'evy's lemma for the uniform distribution (Theorem~\ref{thm: LM} above) yields that $F(x)\, g(x)$ for $u$-most $x$  is close to the $u$-average of $Fg$, which equals the VMF-average of $F$ (where the Lipschitz constant of $f=Fg$ could be a bit worse than that of $F$). 
The set of exceptional $x$ has small $u$-measure, and since $C(D,\kappa)\in[e^{-\kappa},e^\kappa]$ and thus $g(x)\in[e^{-2\kappa},e^{2\kappa}]$, it also has small VMF-measure (larger at most by a factor of $e^{2\kappa}$). Thus, for VMF-most $x$, $F(x)$ is close to VMF$(F)/g(x)$, and thus not constant at all. 
The same argument shows that L\'evy's lemma is violated for any sequence of measures $(\mu_D)_{D\in\NNN}$ on $\SSS(\RRR^D)$ whose density $g_D$ relative to $u$ is bounded uniformly in $D$, has Lipschitz constant bounded uniformly in $D$, but deviates significantly from 1 on a non-negligible set in $\SSS(\RRR^D)$.

For GAP measures the situation is very different. From \eqref{GAPdensity} one can see, for example, that if the eigenvalue $p_{n_2}$ of $\rho=\sum_n p_n |n\rangle\langle n|$ is twice as large as another eigenvalue $p_{n_1}$, then the density \eqref{GAPdensity} at $\psi=|n_2\rangle$ is $2^{D+1}$ times as large as that at $\psi=|n_1\rangle$. Thus, the density and its Lipschitz constant are not (for relevant choices of $\rho$) bounded uniformly in $D$; rather, non-uniform GAP measures become more and more singular with respect to the uniform distribution for large $D$.\hfill$\diamond$
\end{rmk}

\begin{rmk}\label{rmk:gen-can-from-cond}
{\it Generalized canonical typicality from conditional wave function?}
    One might imagine a different strategy of deriving generalized canonical typicality, based on regarding $\psi$ itself as a conditional wave function and using the known fact \cite{GLTZ06b,GLMTZ15} that conditional wave functions are typically $\GAP$ distributed. We could
    introduce a further big system $c$, choose a high-dimensional subspace $\Hilbert_{Rabc}$ in $\Hilbert_{abc}=\Hilbert_a \otimes \Hilbert_b \otimes \Hilbert_c$ so that $\tr_c P_{Rabc}/d_{Rabc}$ coincides with the given $\rho$ on $\Hilbert_a\otimes \Hilbert_b$, and start from a random wave function from $\SSS(\Hilbert_{Rabc})$.
    However, we do not see how to make such a derivation work.
    \hfill$\diamond$
\end{rmk}

\begin{rmk}\label{rmk:counter}
{\it Not every measure does what $\GAP(\rho)$ does.} Generalized canonical typicality as expressed in Theorems \ref{thm:1} and \ref{thm: GCT exp bound} is not true in general if we replace $\GAP(\rho)$ by a different measure: if $\rho$ is a density matrix on $\Hilbert$ and $\mu$ a probability distribution over $\SSS(\Hilbert)$ with density matrix $\rho_\mu=\rho$, then it need not be true for $\mu$-most $\psi$ that $\rho^\psi_a\approx \tr_b \rho$.

Here is a counter-example. Let $\rho=\sum_{n=1}^D p_n |n\rangle \langle n|$ have eigenvalues $p_n$ and eigen-ONB $(|n\rangle)_{n\in\{1,\ldots,D\}}$, and let
\be
\mu=\sum_{n=1}^D p_n \, \delta_{|n\rangle}
\ee
be the measure that is concentrated on the finite set $\{|n\rangle:1\leq n\leq D\}$ and gives weight $p_n$ to each $|n\rangle$. This measure is the narrowest, most concentrated measure with density matrix $\rho$, and thus a kind of opposite of $\GAP(\rho)$, the most spread-out measure with density matrix $\rho$. A random vector $\psi$ with distribution $\mu$ is a random eigenvector $|n\rangle$. What the reduced density matrix $\rho_a^{|n\rangle}$ looks like depends on the vectors $|n\rangle\in\Hilbert=\Hilbert_a\otimes \Hilbert_b$. Suppose that the eigenbasis of $\rho$ is the product of ONBs of $\Hilbert_a$ and $\Hilbert_b$, $|n\rangle=|\ell\rangle_a \otimes |m\rangle_b$; then $\rho_a^{|n\rangle}=\tr_b |n\rangle \langle n| = |\ell\rangle_a\langle \ell|$ (in an obvious notation), so $\rho_a^{|n\rangle}$ is always a pure state and thus far away from $\tr_b \rho= \sum_{\ell,m}p_{\ell m} |\ell\rangle_a\langle\ell|$ if that is highly mixed. Note, however, that if instead of a product basis, we had taken $(|n\rangle)_{n=1\ldots D}$ to be a purely random ONB of $\Hilbert$, then (with overwhelming probability if $d_b\gg 1$) $\rho_a^{|n\rangle}\approx d_a^{-1} I_a$ and thus also $\tr_b \rho$ (which by \eqref{mu-avg-rhopsi} is the $\mu$-average of $\rho_a^\psi$) is close to $d_a^{-1} I_a$, so $\rho_a^{\psi}\approx \tr_b \rho$ for $\mu$-most $\psi$, despite the narrowness of $\mu$.\hfill$\diamond$
\end{rmk}

\begin{rmk}\label{rmk:BEC}{\it Canonical typicality with respect to $\GAP(\rho)$ does not hold for every $\rho$.}
    Let us consider the special case in which $\rho$ has one eigenvalue that is large (e.g., $10^{-1}$), while all others are very small (e.g., $10^{-1000}$). Such a situation occurs for example for $N$-body quantum systems with a gapped ground state $|0\rangle$ at very low temperature, $T$ of order $(\log N)^{-1}$. 
    So call the large eigenvalue $p$ and suppose for definiteness that all other eigenvalues are equal,
    \be
    \rho=p|0\rangle\langle0|+\frac{1-p}{D-1}(I-|0\rangle\langle0|) = p|0\rangle\langle0|+(1-p)\frac{I}{D} + O\Bigl(\frac{1}{D}\Bigr)
    \ee
    with $O(1/D)$ referring to the trace norm and the limit $D\to\infty$. In that case, $\tr\rho^2\approx p^2$ (e.g., $10^{-2}$, while $d_a$ may be $10^{100}$), so the smallness condition \eqref{ineq: cond small} for generalized canonical typicality is strongly violated. To investigate $\rho_a^\psi$, note that any vector $\psi\in\SSS(\Hilbert)$ can be written as $\psi=\cos\theta e^{i\alpha} |0\rangle + \sin\theta |\phi\rangle$ with $\theta\in[0,\pi/2]$, $\alpha\in[0,2\pi)$, and $|\phi\rangle\perp|0\rangle$. If $\psi$ has distribution $\GAP(\rho)$, then $\phi$ has distribution $u_{\SSS(|0\rangle^\perp)}$ and is independent of $\theta$ and $\alpha$, $\alpha$ is independent of $\theta$ and uniformly distributed, and a lengthy computation 
    shows that the distribution of $\theta$ has density 
    \be
    \frac{2(1-p)^2}{p} \frac{\cos\theta}{\sin^5\theta} \exp\Bigl((1-\tfrac{1}{p})\cot^2\theta\Bigr)
    \ee
    as $D\to\infty$. By an error of order $1/\sqrt{D}$, we can replace $\phi$ by a $u_{\SSS(\Hilbert)}$-distributed vector. 
    If $|0\rangle$ factorizes as in $|0\rangle=|0\rangle_a |0\rangle_b$, then $\tr_b\rho = p|0\rangle_a\langle0| + (1-p)(I_a/d_a)+O(1/d_b)$ and $\rho_a^\psi=\cos^2\theta |0\rangle_a \langle0|  + \sin^2\theta (I_a/d_a) + O(1/\sqrt{d_b})$. Since the latter depends on $\theta$ (and thus is not deterministic but has a non-trivial distribution), it follows that $\rho_a^\psi\not\approx \tr_b \rho$ with high probability.\hfill$\diamond$
\end{rmk}

\begin{rmk}{\it Comparison to large deviation theory.} In large deviation theory \cite{Var84}, one studies another version of concentration of measures: one considers a sequence of probability distributions $(\PPP_N)_{N\in\NNN}$ on (say) the real line and studies whether (and at which rate) $\PPP_N\bigl([x,\infty)\bigr)$ tends to 0 exponentially fast as $N\to\infty$ for fixed $x\in\RRR$. Our situation is a bit similar, with the role of $x$ played by $\varepsilon$ in \eqref{expboundepsilon}, and that of $\PPP_N$ by
the distribution of $\|\rho_a^\psi-\tr_b \rho\|_{\tr}$ in $\RRR$ for $\GAP(\rho)$-distributed $\psi$. However, our situation does not quite fit the standard framework of large deviations because we do not necessarily consider a sequence $\rho_N$ of density matrices, but rather a fixed $\rho$ with small $\|\rho\|$. That is why we have provided error bounds in terms of the given $\rho$.\hfill$\diamond$
\end{rmk}

\section{Proofs \label{sec: proofs}}

\subsection{Proof of Remark~\ref{rmk:EEEtrb}}
\label{sec:EEEtrb}

What needs proof here is that also in infinite dimension, the partial trace commutes with the expectation, 
\be
\EEE_\mu\tr_b\pr{\psi}=\tr_b\EEE_\mu\pr{\psi}\,.
\ee
(For $\dim\Hilbert_b<\infty$, $\tr_b$ is a finite sum and thus trivially commutes with $\EEE_\mu$.) So suppose that $\Hilbert_b$ has a countable ONB $(|l\rangle_b)_{ l\in\mathbb{N}}$, and  let $|\phi\rangle_a\in\Hilbert_a$. Then
\begin{subequations}
\begin{align}
    {}_a\langle\phi|\EEE_\mu\tr_b\bigl(\pr{\psi}\bigr)|\phi\rangle_a &= \int_{\mathbb{S}(\Hilbert)} {}_a\langle\phi|\tr_b\bigl(|\psi\rangle\langle\psi|\bigr)|\phi\rangle_a\, \mu(d\psi)\\
    &=\int_{\mathbb{S}(\Hilbert)} \sum_l \bigl|\langle\phi,l|\psi\rangle\bigr|^2\, \mu(d\psi)\\
    &= \sum_l \int_{\mathbb{S}(\Hilbert)} \langle\phi,l|\psi\rangle\langle\psi|l,\phi\rangle \, \mu(d\psi)\\
    &=\sum_l \langle\phi,l|\rho_\mu|l,\phi\rangle\\
    &= {}_a\langle\phi|\tr_b\rho_\mu|\phi\rangle_a,
\end{align}
\end{subequations}
where we used Fubini's theorem in the third and the definition of $\rho_\mu$ in the fourth line.
Since a bounded operator $A$ is uniquely determined by the quadratic form $\phi\mapsto \langle\phi|A|\phi\rangle$, it follows that $\EEE_\mu(\rho_a^\psi)=\tr_b\rho_\mu$.

\subsection{Proof of Theorem~\ref{thm:1}}
\label{sec:PfThm1}

We start with the proof of the polynomial version of generalized canonical typicality and thereby introduce approximation techniques for infinite dimensional Hilbert spaces, which will also be used in the proof of the exponential bounds of Theorem~\ref{thm: GCT exp bound} later on.
For the proof of Theorem~\ref{thm:1} we make use of a result from Reimann \cite{Reimann08}.
Let $(|n\rangle)_{n=1\ldots D}$ be an orthonormal basis of eigenvectors of $\rho$ and $p_1,\dots,p_D$ the corresponding (positive) eigenvalues. 
Reimann used the density of the GAP measure $\GAP(\rho)$ to compute expressions of the form
\begin{align}
    \mathbb{E}(c_j^* c_k c_m^* c_n),
\end{align}
where the expectation is taken with respect to $\GAP(\rho)$ and $c_j = \langle j|\psi\rangle$ are the coordinates of $\psi\in\mathbb{S}(\Hilbert)$ with respect to the orthonormal basis $(|j\rangle)_{j=1\ldots D}$. With the help of these expressions he derived an upper bound for the variance $\Var\langle\psi|A|\psi\rangle$ (also taken with respect to $\GAP(\rho)$) for self-adjoint operators $A:\Hilbert\to\Hilbert$. We show that Reimann's upper bound for $\Var\langle\psi|A|\psi\rangle$ remains essentially valid also for non-self-adjoint $A$ and this bound will be a main ingredient in our proof of Theorem~\ref{thm:1}.

We start by computing the expectation $\mathbb{E}\langle\psi|A|\psi\rangle$ and an upper bound for the variance $\Var \langle\psi|A|\psi\rangle$ for an arbitrary operator $A:\Hilbert\to\Hilbert$, where the expectation and variance are with respect to the measure $\GAP(\rho)$. We closely follow Reimann \cite{Reimann08} who did these computations in the case that $A$ is self-adjoint. We arrive at the same bound for the variance (with the distance between the largest and smallest eigenvalue of $A$ replaced by its operator norm), however, one step in the proof needs to be modified to account for $A$ not being necessarily self-adjoint. Moreover, we show that the expression for $\mathbb{E}\langle\psi|A|\psi\rangle$ and the upper bound for $\Var\langle\psi|A|\psi\rangle$ remain valid if $\Hilbert$ has countably infinite dimension, i.e., if it is separable.

\begin{prop}\label{prop: var}
Let $\rho$ be a density matrix on a separable Hilbert space $\Hilbert$ with positive eigenvalues $p_n$ such that  $p_{\max}=\|\rho\|<1/4$ and let $\dim\Hilbert \geq 4$. For $\GAP(\rho)$-distributed $\psi$ and any bounded operator $A:\Hilbert\to\Hilbert$, 
\begin{align}
    \mathbb{E}\langle\psi|A|\psi\rangle = \tr(A\rho)
\end{align}
and
\begin{align}\label{prop1Var}
    \Var\langle\psi|A|\psi\rangle \leq \frac{\|A\|^2 \tr \rho^2}{1-p_{\max}} \left(1+\frac{4\sqrt{\tr \rho^2} + 2\tr \rho^2}{(1-2p_{\max})(1-3p_{\max})}\right).
\end{align}
\end{prop}

\begin{proof} 
We first assume that $D:=\dim\Hilbert<\infty$.
The formula for the expectation follows immediately from the fact that the density matrix of $\GAP(\rho)$ is $\rho$:
\begin{align}
    \mathbb{E}\langle\psi|A|\psi\rangle = \mathbb{E}\tr(|\psi\rangle\langle\psi|A) = \tr(\mathbb{E}|\psi\rangle\langle\psi|A) = \tr(A\rho).
\end{align}
For a complex-valued random variable $X$ the variance can be computed by
\begin{align}
    \Var X = \mathbb{E}\left[(X-\mathbb{E}X)^*(X-\mathbb{E}X)\right] = \mathbb{E}(X^*X) - \mathbb{E}(X^*)\mathbb{E}(X).
\end{align}
Since the variance of a random variable does not change when a constant is added, we can assume for its computation without loss of generality that $\mathbb{E}\langle\psi|A|\psi\rangle = 0$.
Let $(|n\rangle)_{n=1,\dots,D}$ be an orthonormal basis of $\Hilbert$ consisting of eigenvectors of $\rho$. For $\psi\in\mathbb{S}(\Hilbert)$ we write
\begin{align}
    \langle\psi|A|\psi\rangle = \sum_{l,m} \langle\psi|m\rangle \langle m|A|l\rangle \langle l|\psi\rangle =: \sum_{l,m} c_m^* A_{ml} c_l
\end{align}
with $c_l = \langle l|\psi\rangle$ and $A_{ml} = \langle m|A|l\rangle$. Then for $X=\langle\psi|A|\psi\rangle$ we find that
\begin{align}
    \Var X &= \sum_{l,m,l',m'} A^*_{ml} A_{m'l'} \mathbb{E}(c_l^* c_m  c_{m'}^* c_{l'}).\label{eq: varX}
\end{align}
Reimann \cite{Reimann08} showed that the fourth moments $\mathbb{E}(c_l^* c_m c_{m'}^* c_{l'})$ all vanish except for the two cases $l=m, m'=l'$ and $l=m', m=l'$ and that
\begin{align}
    \mathbb{E}(|c_m|^2 |c_l|^2) = p_m p_l (1+\delta_{ml}) K_{ml},
\end{align}
where 
\begin{align}
    K_{ml} = \int_0^{\infty} (1+xp_m)^{-1} (1+xp_l)^{-1} \prod_{n=1}^D (1+xp_n)^{-1}\, dx.
\end{align}
This implies
\begin{align}
    \Var X &= \sum_{m,l} |A_{ml}|^2 p_m p_l (1+\delta_{ml}) K_{ml} + \sum_{m,m'} A^*_{mm} A_{m'm'} p_m p_{m'}(1+\delta_{mm'}) K_{mm'}\nonumber\\
    &\quad - 2 \sum_m |A_{mm}|^2 p_m^2 K_{mm}\\
    &= \sum_{m,l} \left[|A_{ml}|^2 + A_{mm}^* A_{ll} \right] p_m p_l K_{ml}.\label{eq: var}
\end{align}
Because of $|A_{mm}| \leq \|A\|$ it follows from the computation in \cite{Reimann08} that 
\begin{align}
    \sum_{m,l} A^*_{mm} A_{ll} p_m p_l K_{ml} \leq \frac{2\|A\|^2 \tr\rho^2}{(1-p_{\max})(1-2p_{\max})(1-3p_{\max})} \left(2(\tr \rho^2)^{1/2} + \tr \rho^2\right)\label{ineq: var sum 1}
\end{align}

Moreover, as it was shown in \cite{Reimann08}, $K_{ml} \leq \frac{1}{1-p_{\max}}$ for all $l$ and $m$ and therefore
\begin{align}
    \sum_{m,l} |A_{ml}|^2 p_m p_l K_{ml}\leq \frac{1}{1-p_{\max}} \tr(A^*\rho A \rho).\label{ineq: var sum 2}
\end{align}
Since $A$ is not necessarily self-adjoint, we have to proceed in a different way than Reimann \cite{Reimann08} did to bound this term. To this end we make use of the Cauchy-Schwarz inequality for the trace, i.e. $\tr(B^*C) \leq \sqrt{\tr(B^*B)\tr(C^*C)}$, and the inequality $|\tr(BC)|\leq \|B\| \tr(|C|)$ for any operators $B,C$ \cite[Thm. 3.7.6]{simon}. With these inequalities we have that
\begin{subequations}
\begin{align}
    \tr(A^*\rho A \rho) &\leq \sqrt{\tr(A^*\rho^2 A)\tr(\rho A^*A\rho)}\\ 
    &= \sqrt{\tr(AA^*\rho^2) \tr(A^*A\rho^2) }\\[1mm]
    &\leq \|A\|^2 \tr\rho^2.\label{ineq: tr}
\end{align}
\end{subequations}
Combining \eqref{eq: var}, \eqref{ineq: var sum 1}, \eqref{ineq: var sum 2} and \eqref{ineq: tr} proves the bound for the variance and thus finishes the proof in the finite-dimensional case.

Now suppose that $\Hilbert$ has a countably infinite ONB. The expectation can be computed as before since $\GAP(\rho)(\pr{\psi})=\rho$ remains true in the infinite-dimensional setting \cite{Tum20}. 
For the variance, we approximate $\rho$ by density matrices $\rho_n$, $n\in\mathbb{N}$, of finite rank defined by
\begin{align}
    \rho_n := \sum_{m=1}^{n-1} p_m |m\rangle\langle m| +\Bigl( \sum_{m=n}^{\infty} p_m\Bigr)|n\rangle \langle n|.\label{eq: rho_n}
\end{align}
Then $\|\rho_n-\rho\|_{\tr} \to 0$ as $n\to\infty$, and therefore Theorem~3 in \cite{Tum20} implies that $\GAP(\rho_n) \Rightarrow \GAP(\rho)$ (weak convergence). Note also that from some $n_0$ onwards, $\sum_{m=n}^\infty p_m \leq p_1$ and thus $\|\rho_n\|=p_1=\|\rho\|$. Let $f(\psi):=|\langle\psi|A|\psi\rangle-\tr(A\rho)|^2$ and $f_n(\psi) := |\langle\psi|A|\psi\rangle-\tr(A\rho_n)|^2$. Because of $\tr(A\rho_n)\to\tr(A\rho)$ and therefore $f_n\to  f$ uniformly in $\psi$ it follows that $\GAP(\rho_n)(f_n)-\GAP(\rho_n)(f) \to 0$. Since $f$ is continuous, it follows from the weak convergence of the measures $\GAP(\rho_n)$ that $\GAP(\rho_n)(f) \to \GAP(\rho)(f)$ and therefore altogether that $\GAP(\rho_n)(f_n)\to\GAP(\rho)(f)$. Since, as one easily verifies, $\tr \rho_n^2 \to\tr\rho^2$, the bound for the variance in the finite-dimensional case remains valid in the infinite-dimensional setting.\footnote{A different way to prove that the bound remains valid in the infinite-dimensional setting is the following: Since $\scp{\psi}{A|\psi}$ is a continuous function of $\psi$, it follows from the weak convergence of the measures $\GAP(\rho_n)$ that also the distribution of $\scp{\psi}{A|\psi}$ under $\psi\sim\GAP(\rho_n)$ converges weakly to that under $\psi\sim\GAP(\rho)$ (where the notation $X\sim\mu$ means that the random variable $X$ has distribution $\mu$). 
Since $\tr(A\rho_n) \to \tr(A\rho)$, this does not change if we subtract $\tr(A\rho_n)$ respectively $\tr(A\rho)$ (because the test functions $f$ can equivalently be assumed to be bounded and Lipschitz \cite[Thm.~2.1]{Bill99} and $\langle\psi|A|\psi\rangle$ is Lipschitz), and take the absolute square.
Theorem 3.4 of \cite{Bill99} says that if the distribution of the real random variable $X_n$ converges weakly to that of $X$, then $\EEE |X|\leq \liminf_n \EEE |X_n|$. Thus, the variance of $\scp{\psi}{A|\psi}$ under $\GAP(\rho)$ is bounded by the limit of the bounds for $\rho_n$. Since $\tr\rho_n^2\to\tr\rho^2$, the variance is bounded by the same upper bound as in the finite-dimensional case.}
\end{proof}

\begin{proof}[Proof of Theorem~\ref{thm:1}]
Without loss of generality assume that all eigenvalues of $\rho$ are positive. Proposition~\ref{prop: var} together with Chebyshev's inequality implies for any operator $A$ and any $\varepsilon>0$ that
\begin{subequations}
\begin{align}
    &\GAP(\rho)\left\{\psi\in\mathbb{S}(\Hilbert): \bigl|\langle\psi|A|\psi\rangle - \tr(A\rho) \bigr| >  \varepsilon \right\} \nonumber\\
    &~~~~~\leq \frac{\|A\|^2 \tr \rho^2}{\varepsilon^2(1-p_{\mathrm{max}})} \left(1 + \frac{4\sqrt{\tr\rho^2}+2\tr\rho^2}{(1-2p_{\mathrm{max}})(1-3p_{\mathrm{max}})}\right)\\
    &~~~~~\leq \frac{4\|A\|^2 \tr\rho^2}{3\varepsilon^2}\left(1+8\left(4\sqrt{p_{\max}}+2p_{\max}\right)\right)\\
    &~~~~~\leq \frac{28 \|A\|^2 \tr\rho^2}{\varepsilon^2}.\label{ineq: Gap Cheb}
\end{align}
\end{subequations}
Let $(|l\rangle_a)_{l=1\dots d_a}$ and $(|n\rangle_b)_{n=1\dots d_b}$, where $d_a := \dim\Hilbert_a \in \mathbb{N}$ and $d_b:=\dim\Hilbert_b \in \mathbb{N}\cup \{\infty\}$, be an orthonormal basis of $\Hilbert_a$ and $\Hilbert_b$ respectively. For
\begin{align}
    A^{lm} = \left[|l\rangle_a \langle m|\right] \otimes I_b,\label{eq: op A}
\end{align}
where $I_b$ is the identity on $\Hilbert_b$, we find $\|A^{lm}\|=1$,
\begin{subequations}
\begin{align}
    \langle\psi|A^{lm}|\psi\rangle &= \sum_n\langle\psi| \left(|l\rangle_a\langle m| \otimes |n\rangle_b\langle n|\right) |\psi\rangle\\
    &= {}_a\langle m|\left(\sum_n {}_b\langle n|\psi\rangle \langle\psi|n\rangle_b\right)|l\rangle_a \\
    &= {}_a\langle m|\rho_a^\psi|l\rangle_a\label{eq: psi A psi}
\end{align}
\end{subequations}
and similarly
\begin{subequations}
\begin{align}
    \tr(A^{lm}\rho) &= \sum_{k,n} {}_a\langle k| {}_b\langle n| \left[\left([|l\rangle_a \langle m|] \otimes I_b \right)\rho\right] |k\rangle_a |n\rangle_b \\
    &= {}_a\langle m|\left(\sum_n {}_b\langle n| \rho |n\rangle_b\right)|l\rangle_a \\
    &= {}_a\langle m|\tr_b\rho |l\rangle_a. \label{eq: tr A rho}
\end{align}
\end{subequations}
For any $d_a\times d_a$ matrix $M = (M_{ij})$ it holds that $\|M\|_{\tr} \leq \sqrt{d_a} \|M\|_2$, where $\|M\|_2$ denotes the Hilbert-Schmidt norm of $M$ which is defined by
\begin{align}
    \|M\|_2 = \sqrt{\tr(M^*M)} = \sqrt{\sum_{i,j=1}^{d_a} |M_{ij}|^2},
\end{align}
see, e.g., Lemma~6 in \cite{PSW05}.
Therefore, we have that
\begin{align}
    \|\rho_a^\psi - \tr_b \rho\|^2_{\tr} \leq d_a \sum_{l,m=1}^{d_a} \bigl|{}_a\langle m|\rho_a^\psi-\tr_b\rho|l\rangle_a \bigr|^2
\end{align}
and thus
\begin{subequations}
\begin{align}
    \GAP(\rho)&\left\{\psi\in \mathbb{S}(\Hilbert): \left\|\rho_a^\psi - \tr_b\rho \right\|_{\tr} > d_a^{3/2} \varepsilon\right\}\nonumber\\ 
    &\leq \GAP(\rho)\left\{\psi\in\mathbb{S}(\Hilbert): \sum_{l,m=1}^{d_a}\bigl|{}_a\langle m|\rho_a^\psi-\tr_b\rho|l\rangle_a \bigr|^2 \geq d_a^2 \varepsilon^2\right\}\\
    &\leq \GAP(\rho)\left\{\psi\in\mathbb{S}(\Hilbert): \exists\; l,m : \bigl|{}_a\langle m|\rho_a^\psi-\tr_b\rho|l\rangle_a \bigr| \geq \varepsilon\right\}\\
    &\leq \frac{28d_a^2 \tr\rho^2}{\varepsilon^2},
\end{align}
\end{subequations}
where we used \eqref{ineq: Gap Cheb}, \eqref{eq: psi A psi}, \eqref{eq: tr A rho} and $\|A^{lm}\|=1$ in the last step. By replacing $\varepsilon \to d_a^{-3/2}\varepsilon$ we finally obtain
\begin{align}
    \GAP(\rho)\left\{\psi\in\mathbb{S}(\Hilbert): \|\rho_a^\psi-\tr_b\rho\|_{\tr} > \varepsilon \right\} \leq \frac{28 d_a^5 \tr\rho^2}{\varepsilon^2}.
\end{align}
Setting
\begin{align}
    \delta = \frac{28 d_a^5 \tr\rho^2}{\varepsilon^2}
\end{align}
and solving for $\varepsilon$ gives \eqref{polybounddelta} and thus finishes the proof.
\end{proof}

\subsection{Proof of Theorem~\ref{thm: GCT exp bound}}

The proof of Theorem~\ref{thm: GCT exp bound} follows largely the one of canonical typicality given in \cite{PSW05}; some crucial differences concern our generalization of the L\'evy lemma and the steps needed for covering infinite dimension.

Let $U_a$ be a unitary operator on $\Hilbert_a$. Then the function $f:\mathbb{S}(\Hilbert)\to \mathbb{C}$, $f(\psi) = \tr_a(U_a\rho_a^\psi)=\langle\psi|U_a\otimes I_b|\psi\rangle$ is Lipschitz continuous with Lipschitz constant $\eta \leq 2\|U_a\| = 2$ (see, e.g., Lemma 5 in \cite{PSW05}). By Theorem~\ref{thm: LM for GAP} and Remark~\ref{rmk:LevyfC},
\begin{align}
\GAP(\rho)&\left\{\psi\in\mathbb{S}(\Hilbert):\bigl|\tr_a(U_a\rho_a^\psi)-\GAP(\rho)(\tr_a(U_a \rho_a^\psi)) \bigr| >\varepsilon\right\}\nonumber\\
&\leq 12\exp\left(-\frac{C\varepsilon^2}{8\|\rho\|}\right).
\end{align}
By \eqref{GAPrhopsi},
\begin{align}
    \GAP(\rho)(\tr_a(U_a\rho_a^\psi)) = \tr_a\left(U_a \GAP(\rho)(\rho_a^\psi)\right) = \tr_a\left(U_a \tr_b\rho\right).
\end{align}
Let $(U_a^j)_{j=0}^{d_a^2-1}$ be unitary operators that form a basis for the space of operators on $\Hilbert_a$ such that\footnote{One possible choice is given by 
  \begin{align*}
      U_a^j = \sum_{k=0}^{d_a-1} e^{2\pi i k (j-(j\mathrm{\, mod \,} d_a))/d_a^2} |(k+j)\mbox{ mod } d_a\rangle \langle k|,
  \end{align*}
  where $(|k\rangle)_{k=0\dots d_a-1}$ is an orthonormal basis of $\Hilbert_a$, see \cite{PSW05}.}
  \begin{align}
      \tr_a(U_a^{j*} U_a^k) = d_a \delta_{jk}.
  \end{align}
  Then 
  \begin{align}
     \GAP(\rho)\left\{\psi\in\mathbb{S}(\Hilbert):\exists j: \bigl|\tr_a(U_a^{j}\rho_a^\psi)-\tr_a(U_a^{j}\tr_b\rho) \bigr| >\varepsilon\right\} \leq 12d_a^2 \exp\left(-\frac{C\varepsilon^2}{8\|\rho\|}\right). \label{ineq: GAP all j}  
  \end{align}
  As in \cite{PSW05}, the density matrix $\rho_a^\psi$ can be expanded as
  \begin{align}
      \rho_a^\psi = \frac{1}{d_a}\sum_{j} C_{j}(\rho_a^\psi) U_a^{j},
  \end{align}
  where $C_j(\rho_a^\psi) = \tr_a(U_a^{j*}\rho_a^\psi)$ and \eqref{ineq: GAP all j} becomes
  \begin{align}
    \GAP(\rho)\left\{\psi\in\mathbb{S}(\Hilbert):\exists j: \bigl| C_{j}(\rho_a^\psi)-C_{j}(\tr_b\rho) \bigr| >\varepsilon\right\} \leq 12d_a^2 \exp\left(-\frac{C\varepsilon^2}{8\|\rho\|}\right).  
  \end{align}
  If $|C_{j}(\rho_a^\psi)-C_{j}(\tr_b\rho)|\leq \varepsilon$ for all $j$, then
\begin{subequations}
\begin{align}
      \|\rho_a^\psi-\tr_b\rho\|_{\tr}^2 &\leq d_a\|\rho_a^\psi-\tr_b\rho\|^2_2\\
      &=d_a\left\|\frac{1}{d_a}\sum_{j}\left(C_{j}(\rho_a^\psi)-C_{j}(\tr_b\rho)\right)U_a^{j}\right\|_{2}^2\\
      &= \frac{1}{d_a} \tr_a\left|\sum_{j} \left(C_{j}(\rho_a^\psi)-C_{j}(\tr_b\rho)\right)U_a^{j}\right|^2\\
      &=\sum_{j} \left|C_{j}(\rho_a^\psi)-C_{j}(\tr_b\rho)\right|^2\\
      &\leq d_a^2\varepsilon^2.
\end{align}
\end{subequations}

   This implies that
  \begin{align}
      \GAP(\rho)\left\{\psi\in\mathbb{S}(\Hilbert):\|\rho_a^\psi-\tr_b\rho\|_{\tr} > d_a\varepsilon\right\} \leq 12d_a^2 \exp\left(-\frac{C\varepsilon^2}{8\|\rho\|}\right)
  \end{align}
  and, after replacing $\varepsilon$ by $\varepsilon d_a^{-1}$,
    \begin{align}
      \GAP(\rho)\left\{\psi\in\mathbb{S}(\Hilbert):\|\rho_a^\psi-\tr_b\rho\|_{\tr} > \varepsilon\right\} \leq 12d_a^2 \exp\left(-\frac{C\varepsilon^2}{8d_a^2\|\rho\|}\right).
  \end{align}
  Setting  
  \begin{align}
      \delta = 12d_a^2 \exp\left(-\frac{C\varepsilon^2}{8d_a^2\|\rho\|}\right)
  \end{align}
  and solving for $\varepsilon$ finishes the proof.

\subsection{Proof of Theorem~\ref{thm: LM for GAP}}

The proofs begins with an auxiliary theorem formulated as Theorem~\ref{thm: bound GA} below. For better orientation, we also state the analogous fact about Gaussian distributions as Theorem~\ref{thm: bound G} and  start with quoting its real version:\footnote{The constant in \eqref{Levy-alt} can actually be improved to $1/2$ instead of $2/\pi^2$ \cite[p.~180]{Pis}. But for us it is not important to obtain the optimal constant, and we use a method of proof for Theorem~\ref{thm: bound GA} that yields $2/\pi^2$ in Theorem~\ref{thm: bound G}.}

\begin{lemma}[\cite{MS86}]\label{lemma:LevyGauss}
    Let $F:\mathbb{R}^D \to \mathbb{R}$ be a Lipschitz function with constant $\eta$. Let $X=(X_1,\dots,X_D)$ be a vector of independent (real) standard Gaussian random variables. Then for every $\varepsilon>0$,
    \begin{align}\label{Levy-alt}
        \mathbb{P}\bigl\{|F(X)-\mathbb{E}F(X)|>\varepsilon\bigr\} \leq 2 \exp\left(-\frac{2\varepsilon^2}{\pi^2\eta^2}\right).
    \end{align}
\end{lemma}

Now let $\rho=\sum_{n=1}^D p_n |n\rangle\langle n|$ be a density matrix on the $D$-dimensional Hilbert space $\Hilbert$, and let $Z$ be a random vector in $\Hilbert$ whose distribution is $\mathrm{G}(\rho)$, the Gaussian measure with mean 0 and covariance $\rho$ as defined in Section~\ref{sec:background}; equivalently, $Z=\sum_{n=1}^D Z_n |n\rangle$, where the $Z_n$ are independent complex mean-zero Gaussian random variables with variances
\begin{align}
    \mathbb{E}|Z_n|^2 = p_n \,.
\end{align}
 Then we can write $Z=\sqrt{\rho/2} \tilde{Z}$, where the components $\tilde{Z}_n$ of $\tilde{Z}=\sum_{n=1}^D \tilde{Z}_n |n\rangle$ are $D$ independent complex mean-zero Gaussian random variables with variances $\mathbb{E}|\tilde{Z}_n|^2=2$, which can be in a natural way identified with a vector of $2D$ independent real standard Gaussian variables.

If $F:\Hilbert\to\mathbb{R}$ is Lipschitz with constant $\eta$, then $F\circ \sqrt{\rho/2}: \Hilbert\to \mathbb{R}$ is also Lipschitz with constant $\eta \sqrt{\|\rho\|/2}$. This function can also naturally be considered as a function on $\mathbb{R}^{2D}$ and then an application of Lemma~\ref{lemma:LevyGauss} immediately proves the following theorem: 

\begin{thm}
    \label{thm: bound G}
    Let $\dim\Hilbert<\infty$, let $\rho$ be a density matrix on $\Hilbert$, let $Z$ be a random vector with distribution $\mathrm{G}(\rho)$, and 
    let $F:\Hilbert\to\mathbb{R}$ be a Lipschitz function with Lipschitz constant $\eta$. 
    Then for every $\varepsilon>0$,
    \begin{align}
        \mathbb{P}\bigl\{|F(Z)-\mathbb{E}F(Z)|>\varepsilon\bigr\} \leq 2\exp\left(-\frac{4\varepsilon^2}{\pi^2\eta^2 \|\rho\|}\right).
    \end{align}
\end{thm}

However, instead of using Theorem~\ref{thm: bound G}, we will use Theorem~\ref{thm: bound GA} below, a similar result for the Gaussian adjusted measure $\GA(\rho)$ defined in Section~\ref{sec:background}, which has density $\|\psi\|^2$ relative to $\mathrm{G}(\rho)$. Its proof closely follows the proof of L\'evy's Lemma in \cite{MS86}; for convenience of the reader we provide all the details.

\begin{thm}\label{thm: bound GA}
    Let $\dim\Hilbert<\infty$, let $\rho$ be a density matrix on $\Hilbert$, let $Z$ be a random vector with distribution $\GA(\rho)$, and let $F:\Hilbert\to\mathbb{R}$ be a Lipschitz function with Lipschitz constant $\eta$. 
    Then for every $\varepsilon>0$,
    \begin{align}
        \GA(\rho)\Bigl\{\psi\in\mathbb{S}(\Hilbert):\bigl|F(\psi)-\GA(\rho)(F) \bigr|>\varepsilon\Bigr\} \leq 4\exp\left(-\frac{2\varepsilon^2}{\pi^2\eta^2\|\rho\|}\right). \label{ineq: thm GA bound}
    \end{align}
\end{thm}    

\begin{proof}
    We identify $\Hilbert$ with $\CCC^D$ by means of the ONB $(|n\rangle)_{n=1\ldots D}$. Let $\varphi: \mathbb{R}\to\mathbb{R}$ be a convex function and let $\tilde{Z}=(\tilde{Z_1},\dots,\tilde{Z}_D)$ be a vector with the same distribution as $Z$ but independent of it. With the help of Jensen's inequality and H\"older's inequality we find that
    \begin{subequations}
    \begin{align}
        &\GA(\rho)_{\psi}\left[\varphi(F(\psi)-\GA(\rho)_\phi(F))\right] \nonumber\\[3mm]
        &\qquad\leq \GA(\rho)_\psi \GA(\rho)_\phi\left[\varphi(F(\psi)-F(\phi))\right]\\
        &\qquad= \int_{\Hilbert} \int_{\Hilbert} \varphi(F(\psi)-F(\phi)) \|\psi\|^2 \|\phi\|^2 \, \mathbb{P}(d\psi) \mathbb{P}(d\phi)\\
        &\qquad= \sum_{n,m} \int_{\mathbb{C}^D} \int_{\mathbb{C}^D}\varphi(F(Z)-F(\tilde{Z})) |Z_n|^2 |\tilde{Z}_m|^2 \mathbb{P}(dZ)\mathbb{P}(d\tilde{Z})\\
        &\qquad\leq \sum_{n,m}\left(\mathbb{E}_{(Z,\tilde{Z})}(|Z_n|^4 |\tilde{Z}_m|^4)\mathbb{E}_{(Z,\tilde{Z})}\left(\varphi(F(Z)-F(\tilde{Z}))^2\right)\right)^{1/2},
    \end{align}
    \end{subequations}
    where we use the notation $F(Z)$ and $F(\psi)$ interchangeably.
    We can write $Z_n = \RE Z_n + i \IM Z_n$ where $\RE Z_n$ and $\IM Z_n$ are independent real-valued Gaussian random variables with mean 0 and variance $p_n/2$. Since $\mathbb{E}|\RE Z_n|^2 = p_n/2$ and $\mathbb{E}|\RE Z_n|^4 = 3p_n^2/4$ we obtain
    \begin{align}
        \mathbb{E}|Z_n|^4 = \mathbb{E}|\RE Z_n|^4 + 2 \mathbb{E}|\RE Z_n|^2\mathbb{E}|\IM Z_n|^2 + \mathbb{E}|\IM Z_n|^4 = 2p_n^2\label{eq: exp |Z_n|^4}
    \end{align}
    and therefore
    \begin{align}
        \sum_{n,m}\left(\mathbb{E}_{(Z,\tilde{Z})}(|Z_n|^4 |\tilde{Z}_m|^4)\right)^{1/2} = \sum_{n,m} 2 p_n p_m = 2.
    \end{align}
    We identify $Z$ with the vector $X:=(\RE Z_1,\IM Z_1,\RE Z_2,\dots,\RE Z_D, \IM Z_D)$ of real Gaussian random variables and similarly $\tilde{Z}$ with $Y:= (\RE \tilde{Z}_1,\IM \tilde{Z}_1,\RE \tilde{Z}_2,\dots,\RE \tilde{Z}_D, \IM \tilde{Z}_D)$.
     For each $0\leq \theta\leq \frac{\pi}{2}$ set $X_\theta = X \sin\theta + Y \cos\theta$. One easily sees that the joint distribution of $X$ and $Y$, which is the multivariate normal distribution with mean vector 0 and covariance matrix $\mathrm{diag}(p_1,p_1,\dots,p_D,p_D,,p_1,p_1,\dots,p_D,p_D)/2$, is the same as the joint distribution of $X_\theta$ and $\frac{d}{d\theta} X_\theta = X \cos\theta - Y \sin\theta$ since linear combinations of independent Gaussian random variables are again Gaussian and the entries of the expectation vector and covariance matrix can be easily computed.

     Since $F$ can be approximated uniformly by continuously differentiable functions, we can without loss of generality assume that $F$ is continuously differentiable.

    Let us now assume that $\varphi$ is non-negative. Then $\varphi^2$ is also convex. Then we find with the help of Jensen's inequality that
    \begin{subequations}
    \begin{align}
        \mathbb{E} \varphi(F(Z)-F(\tilde{Z}))^2 &= \mathbb{E}\varphi(F(X)-F(Y))^2\\
        &=\mathbb{E}\left[\varphi\left(\int_0^{\pi/2} \frac{d}{d\theta} F(X_{\theta})\, d\theta\right)^2\right]\\
        &= \mathbb{E}\left[\varphi\left(\int_0^{\pi/2}\left(\nabla F(X_\theta), \frac{d}{d\theta} X_\theta\right)\, d\theta\right)^2 \right]\\
        &\leq \frac{2}{\pi} \mathbb{E}\left[\int_0^{\pi/2}\varphi\left(\frac{\pi}{2}\left(\nabla F(X_\theta),\frac{d}{d\theta} X_\theta\right)\right)^2 d\theta \right]\\
        &= \mathbb{E}\varphi\left(\frac{\pi}{2}\left(\nabla F(X),Y\right)\right)^2,
    \end{align}
    \end{subequations}
    where in the last step we used Fubini's theorem and the fact that the joint distribution of $X_\theta$ and $\frac{d}{d\theta}X_\theta$ is the same as the joint distribution of $X$ and $Y$.

    Let $\lambda\in\mathbb{R}$ and set $\varphi(x)=\exp(\lambda x)$. Then we get
    \begin{subequations}
    \begin{align}
        \mathbb{E}\exp\left[2\lambda(F(X)-F(Y))\right] &\leq \mathbb{E}\exp\left(\lambda\pi\sum_{i=1}^{2D}\frac{\partial F}{\partial x_i}(X) Y_i\right)\\
        &= \mathbb{E}_X \prod_{i=1}^{2D} \mathbb{E}_Y \exp\left(\lambda\pi\frac{\partial F}{\partial x_i}(X) Y_i\right)\\
        &=\mathbb{E} \exp\left(\frac{\lambda^2 \pi^2}{4}\sum_{i=1}^{2D}\left(\frac{\partial F}{\partial x_i}(X)\right)^2p_i\right)\\
        &\leq \mathbb{E}\exp\left(\frac{\lambda^2\pi^2 \|\rho\| \|\nabla F(X)\|^2}{4}\right)\\
        &\leq \exp\left(\frac{\lambda^2\pi^2 \|\rho\| \eta^2}{4}\right).
    \end{align}
    \end{subequations}
    Altogether we obtain
    \begin{align}
        \GA(\rho)\left[\exp(\lambda(F(\psi)-\GA(\rho)(F))) \right] \leq 2\exp\left(\frac{\lambda^2\pi^2 \|\rho\|\eta^2}{8}\right).
    \end{align}
    By Markov's inequality we find that
    \begin{subequations}
    \begin{align}
        &\GA(\rho)\left\{|F(Z)-\GA(\rho)(F)|>\varepsilon\right\} \nonumber\\
        &\qquad = \GA(\rho)\left\{F(Z)- \GA(\rho)(F) > \varepsilon\right\}\nonumber\\
        &~~~~~~~~~~~+ \GA(\rho)\left\{\GA(\rho)(F)-F(Z)>\varepsilon\right\}\\
        &\qquad = \GA(\rho)\left\{\exp(\lambda(F(Z)-\GA(\rho)(F))) > e^{\lambda\varepsilon}\right\}\nonumber\\
        &~~~~~~~~~~~+ \GA(\rho)\left\{\exp(-\lambda(F(Z)-\GA(\rho)(F))) > e^{\lambda\varepsilon}\right\}\\
        &\qquad \leq 4\exp\left(-\lambda\varepsilon + \frac{\lambda^2\pi^2\|\rho\|\eta^2}{8}\right).\label{ineq: proof GA ub}
    \end{align}
    \end{subequations}
    Since $\lambda\in\mathbb{R}$ was arbitrary, we can minimize the right-hand side over $\lambda$. The minimum is attained at $\lambda_{\min} = 4\varepsilon/(\pi^2 \|\rho\|\eta^2)$ and inserting this value in \eqref{ineq: proof GA ub} finally yields \eqref{ineq: thm GA bound}. 
\end{proof}

The last ingredient we need for the proof of Theorem~\ref{thm: LM for GAP} is the following lemma:

\begin{lemma}\label{lem: GA ||psi||<r}
    For all $r>0$ it holds that
\begin{align}
    \GA(\rho)\left\{\|\psi\|<r\right\} \leq \sqrt{2} \exp\left(-\frac{1/2-r^2}{2\|\rho\|}\right).
\end{align}
\end{lemma}
\begin{proof}
    With the help of H\"older's inequality we find that
\begin{subequations}
   \begin{align}
       \GA(\rho)\left\{\|\psi\| < r\right\} &= \sum_n \int_{\Hilbert} |Z_n|^2 \mathbbm{1}_{\{\|\psi\| < r\}} \, \mathbb{P}(d\psi)\\
       &\leq \sum_n \left(\mathbb{E} |Z_n|^4 \mathbb{P}\left(\|\psi\| <r\right)\right)^{1/2}\\
       &= \sqrt{2} \left(\mathbb{P}\left(\|\psi\| < r\right)\right)^{1/2} \label{ineq: GA to P}
   \end{align}
\end{subequations}
   Note that in the third line we used \eqref{eq: exp |Z_n|^4} and that $\sum_n p_n =1$. We can write
   \begin{align}
       \|\psi\|^2 = \sum_n |Z_n|^2 = \sum_{n} p_n |\tilde{Z}_n|^2,
   \end{align}
   where the $\tilde{Z}_n$ are independent complex standard Gaussian random variables. For a random variable $Y$ let $M_Y(t) = \mathbb{E}(e^{tY})$ denote its moment generating function. The Chernoff bound states that for any $a\in\mathbb{R}$,
   \begin{align}
       \mathbb{P}\{Y\leq a\} \leq \inf_{t<0} M_Y(t) e^{-ta}. 
   \end{align}
   Here we thus obtain
   \begin{align}
       \mathbb{P}\left\{\|\psi\| < r\right\}=\mathbb{P}\left\{\|\psi\|^2< r^2\right\} \leq \inf_{t<0} M_{\|\psi\|^2}(t) e^{-tr^2}.
   \end{align}
   We compute
   \begin{align}
       M_{\|\psi\|^2}(t) = \prod_n M_{|\tilde{Z}_n|^2}(p_n t) = \prod_n M_{2(\RE \tilde{Z}_n)^2}\left(\frac{p_n t}{2}\right) M_{2(\IM \tilde{Z}_n)^2}\left(\frac{p_n t}{2}\right).
   \end{align}
   Next note that $2(\RE \tilde{Z}_n)^2$ and $2(\IM \tilde{Z}_n)^2$ are chi-squared distributed random variables with one degree of freedom and that the moment generating function of a random variable $Y$ with distribution $\chi^2_1$ is given by
   \begin{align}
       M_Y(t) = (1-2t)^{-1/2} \quad \mbox{for}\quad t<1/2.
   \end{align}
   Therefore,
   \begin{align}
       M_{\|\psi\|^2}(t) &= \prod_n (1-p_n t)^{-1}
   \end{align}
   and this implies
\begin{subequations}
   \begin{align}
       \mathbb{P}\left\{\|\psi\|<r\right\} &\leq \inf_{t<0} e^{-tr^2} \prod_n (1-p_n t)^{-1}\\
       &= \inf_{t<0} \exp\left(-tr^2 - \sum_n \ln(1-p_n t)\right)\\
       &= \inf_{s>0} \exp\left(sr^2-\sum_n \ln(1+p_n s)\right)\\
       &\leq \exp\left(\frac{r^2}{\|\rho\|} - \sum_n \ln\left(1+\frac{p_n}{\|\rho\|}\right)\right),
   \end{align}
\end{subequations}
   where we chose $s=\|\rho\|^{-1}$ in the last line. Because of
   \begin{align}
       \ln(1+x)\geq \frac{x}{x+1} \geq \frac{x}{2}\quad \mbox{for}\quad 0<x\leq 1
   \end{align}
   we find that
   \begin{align}
    \mathbb{P}\left\{\|\psi\|<r\right\} &\leq \exp\left(\frac{r^2}{\|\rho\|} - \sum_n \frac{p_n}{2\|\rho\|}\right) = \exp\left(-\frac{1/2-r^2}{\|\rho\|}\right).
   \end{align}
   Inserting this into \eqref{ineq: GA to P} finishes the proof.
\end{proof}

\begin{proof}[Proof of Theorem~\ref{thm: LM for GAP}]
We first assume that $D=\dim\Hilbert<\infty$.
      Without loss of generality we can assume that $\GAP(\rho)(f)=0$. Due to the continuity of $f$ it follows that there exists a $\varphi\in\mathbb{S}(\Hilbert)$ such that $f(\varphi)=0$. This implies for all $\tilde{\varphi} \in \mathbb{S}(\Hilbert)$ that
    \begin{align}
        |f(\tilde{\varphi})| = |f(\tilde{\varphi}) - f(\varphi)| \leq \eta \|\tilde{\varphi}-\varphi\| \leq \pi\eta,
    \end{align}
    where we used in the last step that the distance (in the spherical metric) between two points on the unit sphere is bounded by $\pi$. Thus $f$ is bounded by $\pi\eta$.

    Let $0<r<1$ and define $\tilde{f}:\Hilbert \to\mathbb{R}$ by
    \begin{align}
        \tilde{f}(\psi) = \begin{cases} f\left(\frac{\psi}{\|\psi\|}\right)\quad &\mathrm{if}\; \|\psi\| \geq r,\\
        r^{-1} \|\psi\| f\left(\frac{\psi}{\|\psi\|}\right) &\mathrm{if}\; \|\psi\| \leq r.
        \end{cases}
    \end{align}
    For every $\psi,\varphi\in\Hilbert$ such that $\|\psi\|,\|\varphi\| \geq r$ we find that
\begin{subequations}
    \begin{align}
        \left|\tilde{f}(\psi)-\tilde{f}(\varphi)\right| &= \left|f\left(\frac{\psi}{\|\psi\|}\right) - f\left(\frac{\varphi}{\|\varphi\|}\right) \right|\\
        &\leq \eta \left\|\frac{\psi}{\|\psi\|}-\frac{\varphi}{\|\varphi\|}\right\|\\
        &\leq \frac{\eta}{r} \|\psi-\varphi\|,
    \end{align}
\end{subequations}
   where the last inequality follows from
\begin{subequations}
   \begin{align}
       \left\|\frac{\psi}{\|\psi\|} - \frac{\varphi}{\|\varphi\|}\right\|^2 &= 2-\frac{2}{\|\psi\| \|\varphi\|} \RE\langle \psi,\varphi\rangle\\
       &= 2 + 2\RE \langle\psi,\varphi\rangle \left(r^{-2}-\frac{1}{\|\psi\|\|\varphi\|}\right) - 2r^{-2}\RE\langle\psi,\varphi\rangle\\
       &\leq r^{-2}\left(2 \|\psi\| \|\varphi\| - 2 \RE\langle\psi,\varphi\rangle\right)\\
       &\leq r^{-2} \left(\|\psi\|^2+\|\varphi\|^2 - 2\RE\langle\psi,\varphi\rangle\right)\\
       &= r^{-2} \|\psi-\varphi\|^2.
   \end{align}
\end{subequations}
  Thus $\tilde{f}$ is Lipschitz continuous with constant $\eta/r$ on $\{\psi\in\Hilbert: \|\psi\| \geq r\}$.
   
   Now let $\psi,\varphi\in\Hilbert$ such that $\|\psi\|,\|\varphi\| \leq r$ and $\|\varphi\| \leq \|\psi\|$. Then we obtain
\begin{subequations}
   \begin{align}
       \left|\tilde{f}(\psi) - \tilde{f}(\varphi) \right| &= r^{-1} \left|\|\psi\| f\left(\frac{\psi}{\|\psi\|}\right) - \|\varphi\| f\left(\frac{\varphi}{\|\varphi\|}\right) \right|\\
       &\leq r^{-1} \left|\|\psi\| f\left(\frac{\psi}{\|\psi\|}\right) - \|\varphi\| f\left(\frac{\psi}{\|\psi\|}\right) \right| \nonumber\\
       &\quad + r^{-1} \left|\|\varphi\| f\left(\frac{\psi}{\|\psi\|}\right) - \|\varphi\| f\left(\frac{\varphi}{\|\varphi\|}\right) \right|\\
       &\leq \frac{\pi\eta}{r} \bigl|\|\psi\|-\|\varphi\| \bigr| + \frac{\eta}{r} \|\varphi\| \left\|\frac{\psi}{\|\psi\|}-\frac{\varphi}{\|\varphi\|} \right\|\\
       &\leq \frac{5\eta}{r} \|\psi-\varphi\|,
   \end{align}
\end{subequations}
   where the last inequality follows from
\begin{subequations}
   \begin{align}
       \|\varphi\|^2 \left\| \frac{\psi}{\|\psi\|} - \frac{\varphi}{\|\varphi\|} \right\|^2 &= 2\|\varphi\|^2 + 2\RE\langle\psi,\varphi\rangle \left(1-\frac{\|\varphi\|}{\|\psi\|}\right) - 2 \RE\langle\psi,\varphi\rangle\\
       &\leq 2 \|\psi\| \|\varphi\| - 2 \RE\langle\psi,\varphi\rangle\\
       &\leq \|\psi-\varphi\|^2.
   \end{align}
\end{subequations}
   Due to the symmetry of the argument in $\psi$ and $\varphi$, one finds the same estimate in the case that $\|\psi\| \leq \|\varphi\|$ and we conclude that $\tilde{f}$ is Lipschitz continuous with constant $5\eta/r$ on $\{\psi\in\Hilbert: \|\psi\| \leq r\}$.

   Finally, let $\psi,\varphi \in \Hilbert$ such that $\|\psi\| \leq r$ and $\|\varphi\| \geq r$ and define $\gamma: [0,1] \to \Hilbert, \gamma(t)=(1-t)\psi+ t\varphi$. Then there exists a $t_0\in [0,1]$ such that $\|\gamma(t_0)\|=r$ and 
   \begin{align}
       \|\psi-\gamma(t_0)\| &= t_0 \|\psi-\varphi\| \leq \|\psi-\varphi\|,\\
       \|\gamma(t_0) - \varphi\| &= (1-t_0) \|\psi-\varphi\| \leq \|\psi-\varphi\|.
   \end{align}
   Therefore, we find that
\begin{subequations}
   \begin{align}
       \left|\tilde{f}(\psi) - \tilde{f}(\varphi)\right| 
       &= \left|r^{-1}\|\psi\| f\left(\frac{\psi}{\|\psi\|}\right) - f\left(\frac{\varphi}{\|\varphi\|}\right) \right|\\
       &\leq r^{-1} \left|\|\psi\| f\left(\frac{\psi}{\|\psi\|}\right) - \|\gamma(t_0)\| f\left(\frac{\gamma(t_0)}{\|\gamma(t_0)\|}\right) \right| \nonumber\\
       &\quad + \left|f\left(\frac{\gamma(t_0)}{\|\gamma(t_0)\|}\right) - f\left(\frac{\varphi}{\|\varphi\|}\right)\right|\\
       &\leq \frac{5\eta}{r} \|\psi-\gamma(t_0)\| + \frac{\eta}{r} \|\gamma(t_0)-\varphi\|\\
       &\leq \frac{6\eta}{r} \|\psi-\varphi\|.
   \end{align}
\end{subequations}
   We conclude that $\tilde{f}$ is Lipschitz continuous with Lipschitz constant $6\eta/r$.

    Using the definition of $\tilde{f}$, we find that
\begin{subequations}
   \begin{align}
       \GAP(\rho)\left\{|f(\psi)|>\varepsilon\right\} &= \GA(\rho)\left\{\left|f\left(\frac{\psi}{\|\psi\|}\right) \right|>\varepsilon\right\}\\
       &\leq \GA(\rho)\left\{\left|f\left(\frac{\psi}{\|\psi\|}\right) \right| > \varepsilon \mbox{ and } \|\psi\| \geq r\right\} + \GA(\rho)\left\{\|\psi\| < r\right\}\\
       &= \GA(\rho)\left\{\left|\tilde{f}(\psi) \right|>\varepsilon \mbox{ and } \|\psi\|\geq r\right\} + \GA(\rho)\left\{\|\psi\| < r\right\}\\
       &\leq \GA(\rho)\left\{\left|\tilde{f}(\psi) \right|>\varepsilon\right\} + \GA(\rho)\left\{\|\psi\| <r\right\}\\
       &\leq \GA(\rho)\left\{\left|\tilde{f}(\psi)- \GA(\rho)(\tilde{f})\right| > \varepsilon - |\GA(\rho)(\tilde{f})|\right\} + \GA(\rho)\left\{\|\psi\|<r\right\}.\label{ineq: GAP eps}
   \end{align}
\end{subequations}
   By Lemma~\ref{lem: GA ||psi||<r}, the second term can be bounded by $\sqrt{2}\exp(-(1/2-r^2)/2\|\rho\|)$. In order to estimate the first term in \eqref{ineq: GAP eps}, we first derive an upper bound for $|\GA(\rho)(\tilde{f})|$. We compute
   \begin{align}
       \GA(\rho)(\tilde{f}) &= \int_{\{\|\psi\| <r\}} r^{-1} \|\psi\| f\left(\frac{\psi}{\|\psi\|}\right)\, \GA(\rho)(d\psi) 
       + \int_{\{\|\psi\| \geq r\}} f\left(\frac{\psi}{\|\psi\|}\right) \, \GA(\rho)(d\psi) \\
       &= \underbrace{\int_\Hilbert f\left(\frac{\psi}{\|\psi\|}\right)\, \GA(\rho)(d\psi)}_{=\GAP(\rho)(f)=0} + \int_{\{\|\psi\| <r\}} r^{-1} \|\psi\| f\left(\frac{\psi}{\|\psi\|}\right) - f\left(\frac{\psi}{\|\psi\|}\right)\, \GA(\rho)(d\psi)
   \end{align}
   and so we obtain, again by Lemma~\ref{lem: GA ||psi||<r},
   \begin{align}
       |\GA(\rho)(\tilde{f})| \leq \pi\eta \,\GA(\rho)\left\{\|\psi\| < r\right\} \leq 5\eta \exp\left(-\frac{1/2-r^2}{2\|\rho\|}\right).
   \end{align}
   This implies with the help of Theorem~\ref{thm: bound GA} that
\begin{subequations}
   \begin{align}
       \GA(\rho)&\left\{\left|\tilde{f}(\psi)- \GA(\rho)(\tilde{f})\right| > \varepsilon - |\GA(\rho)(\tilde{f})|\right\}\\ 
       &\leq \GA(\rho)\left\{\left| \tilde{f}(\psi) - \GA(\rho)(\tilde{f})\right| > \varepsilon - 5\eta \exp\left(-\frac{1/2-r^2}{2\|\rho\|}\right)\right\}\\
       &\leq 4\exp\left(-\frac{r^2(\varepsilon-5\eta\exp(-(1/2-r^2)/2\|\rho\|))^2}{18\pi^2\eta^2 \|\rho\|}\right),
   \end{align}    
\end{subequations}
   provided that $\varepsilon > 5\eta\exp(-(1/2-r^2)/2\|\rho\|)$. Altogether we arrive at
   \begin{align}
       \GAP(\rho)\left\{|f(\psi)|>\varepsilon\right\} \leq 4\exp\left(-\frac{r^2(\varepsilon-5\eta\exp(-(1/2-r^2)/2\|\rho\|))^2}{18\pi^2\eta^2 \|\rho\|}\right) + \sqrt{2}\exp\left(-\frac{1/2-r^2}{2\|\rho\|}\right).
   \end{align}
   Choosing $r=1/2$ we obtain
   \begin{align}\label{estimate1}
       \GAP(\rho)\left\{|f(\psi)|>\varepsilon\right\} \leq 4 \exp\left(-\frac{(\varepsilon-5\eta\exp(-1/8\|\rho\|))^2}{72\pi^2\eta^2 \|\rho\|}\right) + \sqrt{2} \exp\left(-\frac{1}{8\|\rho\|}\right).
   \end{align}
   We can assume without loss of generality that 
   \be\label{estimate3}\varepsilon<\pi\eta
   \ee
   because otherwise the left-hand side of \eqref{LMGAP} vanishes: indeed, the distance between any two points on the sphere is at most $\pi$, so their $f$ values can differ at most by $\pi\eta$, and for the same reason $f(\psi)$ can differ from its average relative to any measure by at most $\pi\eta$. 
   
   Likewise, we can assume without loss of generality that 
   \be\label{estimate2}
    \varepsilon \geq 10\eta\exp(-1/8\|\rho\|)
   \ee
   because otherwise the right-hand side of \eqref{LMGAP} is greater than 1: indeed, for $\varepsilon < 10\eta \exp(-1/8\|\rho\|)$,
    \begin{align}
     6 \exp\left(-\frac{\varepsilon^2}{288\pi^2\eta^2 \|\rho\|}\right)\geq 6\exp\left(-\frac{25 \exp(-1/4\|\rho\|)}{72\pi^2\|\rho\|}\right)>1.
     \end{align}
   As a consequence of \eqref{estimate3} and \eqref{estimate2}, the first exponent in \eqref{estimate1} is greater than the second, so
   \begin{align}
       \GAP(\rho)\left\{|f(\psi)|>\varepsilon\right\} 
       &\leq 6 \exp\left(-\frac{(\varepsilon-5\eta\exp(-1/8\|\rho\|))^2}{72\pi^2\eta^2 \|\rho\|}\right)\\
        &\leq 6 \exp\left(-\frac{\varepsilon^2}{288\pi^2\eta^2 \|\rho\|}\right) \label{ineq: bound GAP}
   \end{align}
   by \eqref{estimate2}. This finishes the proof in the finite-dimensional case.
   
    Now suppose that $\Hilbert$ has a countably infinite ONB. Consider the density matrices $\rho_n$ defined as in \eqref{eq: rho_n}. Let $\varepsilon'>0$. Because the set
\begin{align}
    A_\varepsilon := \{\psi\in\mathbb{S}(\Hilbert): |f(\psi)|>\varepsilon\}
\end{align}
is open in $\SSS(\Hilbert)$, it follows from the weak convergence of the measures $\GAP(\rho_n)$ to $\GAP(\rho)$ by the ``portmanteau theorem'' \cite[Thm.~2.1]{Bill99} that 
\begin{align}
    \GAP(\rho)(A_\varepsilon) \leq \liminf_{n\to\infty} \GAP(\rho_n)(A_\varepsilon) \leq \GAP(\rho_N)(A_\varepsilon) + \varepsilon'
\end{align}
for some large enough $N\in\mathbb{N}$ with $N\geq n_0$. Recall that $n_0\in\mathbb{N}$ is chosen such that $\|\rho_n\|=\|\rho\|$ for all $n\geq n_0$. Let $\Hilbert_N := \mathrm{span}\{|n\rangle: n=1,\dots,N\}$. Then, since $\rho_N$ is a density matrix on $\Hilbert_N$ and $\GAP(\rho_N)$ is concentrated on $\Hilbert_N$, it follows with what we have already proven in the finite-dimensional case that
\begin{subequations}
\begin{align}
    \GAP(\rho_N)\bigl\{\psi\in\mathbb{S}(\Hilbert): |f(\psi)|>\varepsilon\bigr\} &= \GAP(\rho_N)\bigl\{\psi\in\mathbb{S}(\Hilbert_N): |f(\psi)|>\varepsilon\bigr\}\\
    &\leq 6 \exp\left(-\frac{C\varepsilon^2}{\eta^2 \|\rho_N\|}\right),
\end{align}
\end{subequations}
where $C=\frac{1}{288\pi^2}$. Noting that $\|\rho_N\|=\|\rho\|$ and that $\varepsilon'>0$ was arbitrary, we can altogether conclude that
\begin{align}
    \GAP(\rho)\Bigl\{\psi\in\mathbb{S}(\Hilbert): \bigl|f(\psi)-\GAP(\rho)(f)\bigr|>\varepsilon\Bigr\} \leq 6\exp\left(-\frac{C\varepsilon^2}{\eta^2 \|\rho\|}\right),
\end{align}
i.e., the bound \eqref{ineq: bound GAP} remains true in the infinite-dimensional setting.
\end{proof}

\subsection{Proofs of Corollaries~\ref{cor: dyn typ}, \ref{cor: dyn can typ}, \ref{cor:cond}\label{proof: cor1, cor2, cor3}}

\begin{proof}[Proof of Corollary~\ref{cor: dyn typ}]
    As already noted before Corollary~\ref{cor: dyn typ}, the first inequality follows immediately from Corollary~\ref{cor:LevyB} by inserting $U_t^* B U_t$ for $B$.

    For the proof of the second inequality we define
    \begin{align}
        Y_t := \left|\langle\psi_t|B|\psi_t\rangle - \tr(\rho_t B) \right|.
    \end{align}
    Then, for every $s> 0$ we find that
    \begin{align}
        \GAP(\rho)\Bigl\{\psi\in\mathbb{S}(\Hilbert): e^{sY_t}>e^{s\varepsilon}\Bigr\} \leq 12 \exp\left(-\frac{\tilde{C}\varepsilon^2}{\|B\|^2 \|\rho\|}\right),
    \end{align}
    i.e., with $\delta := e^{s\varepsilon}$,
    \begin{align}
        \GAP(\rho)\Bigl\{\psi\in\mathbb{S}(\Hilbert): e^{sY_t} > \delta\Bigr\} \leq 12\exp\left(-\frac{\tilde{C}}{\|B\|^2\|\rho\|} \frac{\ln(\delta)^2}{s^2}\right).
    \end{align}
    This implies
    \begin{subequations}
    \begin{align}
        \GAP(\rho)\left(e^{sY_t}\right) &\leq \sum_{n=0}^\infty (n+1)\; \GAP(\rho)\Bigl\{\psi\in\mathbb{S}(\Hilbert): e^{s Y_t} \in (n,n+1]\Bigr\} \\
        &\leq 1+12 \sum_{n=1}^\infty (n+1) \exp\left(-\frac{\tilde{C}\ln(n)^2}{\|B\|^2 \|\rho\|s^2}\right)\\
        &= 1+ 12\sum_{n=1}^\infty (n+1) n^{-\frac{\tilde{C}\ln(n)}{\|B\|^2\|\rho\| s^2}}.
    \end{align}
    \end{subequations}
    With $a:=\frac{\tilde{C}}{\|B\|^2\|\rho\| s^2}$ and assuming that $a\leq 1$ we obtain
    \begin{subequations}
        \begin{align}
            \GAP(\rho)\left(e^{sY_t}\right) &\leq 1+12 \sum_{n=1}^{\lfloor e^{5/2a}\rfloor} (n+1) + 12 \sum_{n=\lceil e^{5/2a}\rceil}^{\infty} (n+1) \frac{1}{n^{5/2}}\\
            &\leq 1+ 6 e^{\frac{5}{2a}} \left(e^{\frac{5}{2a}}+3\right) + 12\\
            &= 13 + 18e^{\frac{5}{2a}} + 6e^{\frac{5}{a}}\\
            &\leq 9 e^{\frac{5}{a}}.
        \end{align}
    \end{subequations}
    An application of Jensen's inequality and Fubini's theorem shows that
    \begin{subequations}
    \begin{align}
        \GAP(\rho)\left(\exp\left(\frac{1}{T}\int_0^T Y_t\, dt\; s\right)\right) &\leq \GAP(\rho)\left(\frac{1}{T}\int_0^T e^{Y_t s}\, dt\right)\\
        &= \frac{1}{T}\int_0^T \GAP(\rho)\left(e^{Y_t s}\right)\, dt\\
        &\leq 9e^{5/a}.
    \end{align}
    \end{subequations}
    With the help of Markov's inequality we find that
    \begin{align}
        \GAP(\rho)\Bigl\{\psi\in\mathbb{S}(\Hilbert): \frac{1}{T}\int_0^T Y_t\, dt > \varepsilon\Bigr\} \leq 9e^{5/a} e^{-\varepsilon s}.
    \end{align}
    and choosing $s:=\frac{\varepsilon \tilde{C}}{6\|B\|^2 \|\rho\|}$ yields the desired bound provided that $\varepsilon>0$ and $a\leq 1$, i.e., $\|\rho\|~\leq~\frac{\tilde{C}\varepsilon^2}{36\|B\|^2}$. However, since the bound becomes trivial for $\|\rho\|>\frac{\tilde{C}\varepsilon^2}{36\|B\|^2}$, this assumption on $\|\rho\|$ can be dropped. Moreover, note that the bound is also trivial if $\varepsilon=0$.
\end{proof}

\begin{proof}[Proof of Corollary~\ref{cor: dyn can typ}]
    Let
    \begin{align}
        Z_t := \|\rho_a^{\psi_t}-\tr_b\rho_t\|_{\tr}.
    \end{align}
    It follows from the equivariance of $\rho\mapsto \GAP(\rho)$ and Remark~\ref{rmk: exp bound} that
    \begin{subequations}
    \begin{align}
        \GAP(\rho)\Bigl\{\psi\in\mathbb{S}(\Hilbert): Z_t > \varepsilon\Bigr\} 
        &= \GAP(\rho_t)\Bigl\{\psi_t\in\mathbb{S}(\Hilbert): Z_t > \varepsilon\Bigr\} \\
        &\leq 12 d_a^2 \exp\left(-\frac{\tilde{C}\varepsilon^2}{d_a^2\|\rho\|}\right).
    \end{align}
    \end{subequations}
    The rest of the proof now follows along the same lines as the proof of Corollary~\ref{cor: dyn typ}.
\end{proof}

\begin{proof}[Proof of Corollary~\ref{cor:cond}.]
    Choose $\psi$ and $B$ independently with the distributions mentioned. By Theorem~2 of \cite{GLMTZ15} (which requires that $d_b\geq d_a$ and $d_b\geq 4$), we have that
    \be
    \bigl|\mathrm{Born}_a^{\psi,B}(f) - \GAP(\rho_a^\psi)(f)\bigr|< \varepsilon/2
    \ee
    with probability $\geq 1 - 16\|f\|^2_{\infty}/\varepsilon^2 d_b \geq 1-\delta/2$ for $d_b\geq 32\|f\|^2_\infty/\varepsilon^2 \delta$. By Lemma 5 of \cite{GLMTZ15}, there is $r=r(\varepsilon,d_a,f)>0$ such that
    \be
    \bigl|\GAP(\rho_a^\psi)(f) - \GAP(\tr_b \rho)(f) \bigr|<\varepsilon/2
    \ee
    whenever $\|\rho_a^\psi-\tr_b\rho\|_{\tr} < r$. By Theorem~\ref{thm: GCT exp bound} in the form \eqref{expboundepsilon}, the latter condition is fulfilled with probability $\geq 1-12d_a^2\exp(-\tilde{C}r^2/d_a^2\|\rho\|)\geq 1-\delta/2$ for $\|\rho\|\leq p:= \tilde{C}r^2/d_a^2\ln(24d_a^2/\delta)$.
    Now \eqref{typcond} follows. 
\end{proof}

\subsection{Further Explanations to Remark~\ref{rmk:comparison}\label{sec: rmk8 explanations}}

As discussed after Theorem~\ref{thm:4}, applying Theorem~\ref{thm:1} to $\rho=\rho_R$ yields the worse factor $d_a^{2.5}$ instead of $d_a^2$. Here we want to give some details why in this special case of Theorem~\ref{thm:1}, slightly better  bounds can be obtained.

First suppose that $\Hilbert_R = \Hilbert$. Similarly to the proof of Theorem~\ref{thm:1} one finds that
\begin{align}
    u\bigl\{\psi\in\mathbb{S}(\Hilbert): \|\rho_a^\psi-\tr_b\rho\|_{\tr}>\varepsilon\bigr\} \leq \frac{d_a^3}{\varepsilon^2} \sum_{l,m}\Var\langle\psi|A^{lm}|\psi\rangle,
\end{align}
where $A^{lm} = |l\rangle_a\langle m|\otimes I_b$ and $(|l\rangle_a)_{l=1\dots d_a}$ is an orthonormal basis of $\Hilbert_a$. Instead of bounding the sum by $d_a^2$ times a uniform bound on the variances $\Var\langle\psi|A^{lm}|\psi\rangle$, one can now make use of the fact that for uniformly distributed $\psi\in\mathbb{S}(\Hilbert)$, the second and fourth moments of the coefficients $c_l$ of $\psi$ in an orthonormal basis $(|n\rangle)_{n=1\dots D}$ of eigenvectors of $\rho$ can be computed explicitly. More precisely, they satisfy
\begin{align}
    \mathbb{E}(|c_n|^2) = \frac{1}{D}, \quad \mathbb{E}(|c_n|^2|c_k|^2) = \frac{1+\delta_{nk}}{D(D+1)},
\end{align}
and all other second and fourth moments vanish, see e.g. \cite[App. A.2 and C.1]{GMM04}. With this we find that
\begin{subequations}
\begin{align}
    \Var\langle\psi|A^{lm}|\psi\rangle
    &= \sum_{k,n} |A_{kn}^{lm}|^2 \frac{1+\delta_{kn}}{D(D+1)} + \sum_{k,n} A^{lm*}_{kk} A_{nn}^{lm} \frac{1+\delta_{kn}}{D(D+1)}\nonumber\\ 
    &\qquad - \sum_n |A_{nn}^{lm}|^2 \frac{2}{D(D+1)} - \tr(A^{lm}\rho)^2\\
    &= \frac{\tr(A^{lm*}A^{lm})}{D(D+1)} - \frac{\left|\tr(A^{lm}\rho)\right|^2}{D+1}\\
    &\leq \frac{\tr(A^{lm*}A^{lm})}{D(D+1)}.
\end{align}
\end{subequations}
Next note that 
\begin{align}
    \sum_{l,m} \tr(A^{(lm)*}A^{(lm)}) = d_a \sum_{l} \tr(|l\rangle_a\langle l|\otimes I_b) = d_a D
\end{align}
and therefore
\begin{align}
    u\left\{\psi\in\mathbb{S}(\Hilbert): \|\rho_a^\psi-\tr_b\rho\|_{\tr}>\varepsilon\right\} \leq \frac{d_a^4}{\varepsilon^2 D}. 
\end{align}
If $\Hilbert_R \neq \Hilbert$ is a subspace of $\Hilbert$, then this bound remains valid after replacing $\rho$ by $P_R/d_R$, $u$ by $u_R$, $\Hilbert$ by $\Hilbert_R$ and $D$ by $d_R$. This follows immediately from the previous computations after noting that 
\begin{align}
    \sum_{l,m}\tr(A^{(lm)*}P_RA^{(lm)}P_R) \leq \sum_{l,m} \tr(A^{(lm)*}P_R A^{(lm)}) = d_a \sum_l \tr((|l\rangle_a\langle l|\otimes I_b P_R) = d_a d_R.
\end{align}
Setting $\delta:= d_a^4/(\varepsilon^2 d_R)$ and solving for $\varepsilon$ finally gives Theorem~\ref{thm:4}.

In \cite{PSW05,PSW06}, Theorem~\ref{thm:PSW} was stated in a slightly different form; more precisely, there it was shown that for every $\varepsilon>0$,
\begin{align}
    u_R \Biggl\{ \psi \in \SSS(\Hilbert_R): \bigl\|\rho_a^\psi -
    \tr_b \rho_R  \bigr\|_{\tr} > \varepsilon + \sqrt{d_a \tr(\tr_a \rho_R)^2}
   \Biggr\} \leq 4 \exp\Bigl(-\frac{d_R\varepsilon^2}{18\pi^3}\Bigr).
\end{align}
We now show how this implies the bound in Theorem~\ref{thm:PSW}. By setting $\delta := 4 \exp(-d_R\varepsilon^2/(18\pi^3))$ and solving for $\varepsilon$, we obtain
\begin{align}
    \varepsilon = \sqrt{\frac{18\pi^3}{d_R}\ln(4/\delta)}.
\end{align}
With this and $\tr(\tr_a\rho_R)^2 \leq d_a/d_R$ we obtain
\begin{align}
     u_R \Biggl\{ \psi \in \SSS(\Hilbert_R): \bigl\|\rho_a^\psi -
  \tr_b \rho_R  \bigr\|_{\tr} \leq  \sqrt{\frac{18\pi^3}{d_R}\ln(4/\delta)}+ \sqrt{d_a^2/d_R}
  \Biggr\} \geq 1-\delta\,.
\end{align}
The first square root dominates as soon as
\begin{align}
    \delta < 4\exp\left(-d_a^2/(18\pi^3)\right),
\end{align}
which we can, of course, assume without loss of generality since otherwise we would have $\delta >1$ and then the lower bound on the probability would be trivial. This immediately implies \eqref{cantyp}.

\section{Summary and Conclusions \label{sec: concl}}

Typicality theorems assert that, for big systems, some condition is true of \emph{most} points, or here, most wave functions. The word ``most'' usually refers to a uniform distribution $u$ (say, over the unit sphere $\SSS(\Hilbert_R)$ in some Hilbert subspace $\Hilbert_R$), but here we use the GAP measure as the natural analog of the uniform distribution in cases with given density matrix $\rho$. Since the GAP measure for $\rho=\rho_\can$ is the thermal equilibrium distribution of wave functions, our typicality theorems can be understood as expressing a kind of equivalence of ensembles between a micro-canonical ensemble of wave functions ($u_{\SSS(\Hilbert_{\mc})}$) and a canonical ensemble of wave functions ($\GAP(\rho_\can)$). Yet, our results apply to arbitrary $\rho$.

The key mathematical step is the generalization of L\'evy's lemma to GAP measures, that is, of the concentration of measure on high-dimensional spheres.
The fact that the pure states of a quantum system are always the points on a sphere then allows us to deduce very general typicality theorems from this kind of concentration of measure. In particular, these typicality statements are largely independent of the properties of the Hamiltonian and require only that many dimensions participate in $\rho$.

Specifically, some of these statements concern a bi-partite quantum system $a\cup b$, where $b$ is macroscopically large. We have shown that for most $\psi$ from the $\GAP(\rho)$ ensemble, the reduced density matrix $\rho_a^\psi$ is close to its average $\tr_b \rho$ assuming that the largest eigenvalue (Theorem~\ref{thm: GCT exp bound}) or at least the average eigenvalue (Theorem~\ref{thm:1}) of $\rho$ is small. That is, we have established an extension of canonical typicality to GAP measures. This family of measures is particularly natural in this context because it arises anyway in the context of bi-partite systems as the typical asymptotic distribution of the conditional wave function \cite{GLTZ06b,GLMTZ15}, a fact extended further in Corollary~\ref{cor:cond}.

Another important application of concentration-of-measure of GAP yields (Corollary~\ref{cor:LevyB}) that for any observable $B$, most $\psi$ from the $\GAP(\rho)$ ensemble have nearly the same Born distribution (when suitably coarse grained). Moreover (Corollaries~\ref{cor: dyn typ} and \ref{cor: dyn can typ}), if the initial wave function $\psi_0$ is $\GAP(\rho)$-distributed, then for any unitary time evolution the whole curves $t\mapsto\langle\psi_t|B|\psi_t\rangle$ and $t\mapsto \rho_a^{\psi_t}$ are nearly deterministic (and given by $\tr (B \rho_t)$ and $\tr_b \rho_t$).

All these results contribute different aspects to the picture of how an individual, closed quantum system in a pure state can display thermodynamic behavior \cite{vonNeumann29,Schroe52, Deutsch91,Srednicki94,Tasaki98,GM03,GMM04,PSW06,Reimann08b, BG09,GLMTZ09,GLMTZ10,GLTZ10,Short11,SF12,GHT13,GHT15,Reimann2015, GogEis16,BRGSR18,Reimann2018a,Reimann2018b,RG20,TTV22-physik,SWGW22}, and thus help clarify the role of ensembles as defining a concept of typicality, while thermal density matrices arise from partial traces. 

In sum, our results describe simple relations between the following concepts: reduced density matrix, many participating dimensions, and GAP measures. That is, if many dimensions participate in $\rho$, then for $\GAP(\rho)$-most $\psi$, the reduced density matrix $\rho_a^\psi$ is nearly independent of $\psi$. 

\bigskip

\noindent \textbf{Acknowledgments.}
We thank Tristan Benoist, Andreas Deuchert, Marius Lemm, and Martin M\"ohle for helpful discussion. C.V.\ acknowledges financial support by the German Academic Scholarship Foundation. S.T.\ acknowledges financial support by the Deutsche Forschungsgemeinschaft (DFG, German Research Foundation) -- TRR 352 -- Project-ID 470903074.
\\
\\
\noindent\textbf{Data Availability Statement.}
Data sharing is not applicable to this article as no datasets were generated or analysed.
\\
\\
\noindent\textbf{ Conflict of Interest Statement.}
The authors have no conflicts of interest.

\bibliographystyle{plainurl}
\bibliography{LiteratureGCT.bib}

\end{document}